\documentclass{amsart}

\usepackage{amssymb}
 \usepackage{amsthm}

\usepackage{multirow}
\usepackage[boxruled,vlined]{algorithm2e}
\usepackage{tikz}
\usetikzlibrary{matrix}
\usepackage{subfigure}

 \newtheorem{theorem}{Theorem}[section]

 \newtheorem{example}[theorem]{Example}
 \theoremstyle{definition}
 
 \theoremstyle{remark}
 \newtheorem{remark}[theorem]{Remark}

\begin{document}


\title[Lightweight LCP Construction for Very Large Collections of Strings]{Lightweight LCP Construction \\for Very Large Collections of Strings}
\thanks{\textcopyright  2016. This manuscript version is made available under the CC-BY-NC-ND 4.0 license http://creativecommons.org/licenses/by-nc-nd/4.0/ \\
The final version of this manuscript is in press in Journal of Discrete Algorithms. DOI: 10.1016/j.jda.2016.03.003.}


\author[A.J. Cox]{Anthony J. Cox}
\address[A.J. Cox]{Illumina Cambridge Ltd., United Kingdom}
\email{acox@illumina.com}
\author[F. Garofalo]{Fabio~Garofalo} 
\address[F. Garofalo]{University of Palermo, Dipartimento di Matematica e Informatica, ITALY.}
\email{garofalo{\_}uni@yahoo.it}
\author[G. Rosone]{Giovanna~Rosone}
\address[G. Rosone]{University of Pisa, Dipartimento di Informatica, ITALY, \\
	University of Palermo, Dipartimento di Matematica e Informatica, ITALY.}
\email{giovanna.rosone@unipi.it}
\author[M. Sciortino]{Marinella~Sciortino} 
\address[M. Sciortino]{University of Palermo, Dipartimento di Matematica e Informatica, ITALY.}
\email{marinella.sciortino@unipa.it}

\begin{abstract}
The longest common prefix array is a very advantageous data structure that, combined with the suffix array and the Burrows-Wheeler transform, allows to efficiently compute some combinatorial properties of a string useful in several applications, especially in biological contexts. Nowadays, the input data for many problems are big collections of strings, for instance the data coming from ``next-generation'' DNA sequencing (NGS) technologies. In this paper we present the first lightweight algorithm (called \texttt{extLCP}) for the simultaneous computation of the longest common prefix array and the Burrows-Wheeler transform of a very large collection of strings having any length. The computation is realized by performing disk data accesses only via sequential scans, and the total disk space usage never needs more than twice the output size, excluding the disk space required for the input.
Moreover, \texttt{extLCP} allows to compute also the suffix array of the strings of the collection, without any other further data structure is needed. Finally, we test our algorithm on real data and compare our results with another tool capable to work in external memory on large collections of strings.
\end{abstract}





\maketitle


\def\cbwt(#1){\texttt{bwt}(#1)}
\def\cebwt(#1){\texttt{c-ebwt}(#1)}
\def\lbwt(#1){\texttt{bwt}(#1)}
\def\lbwt{\texttt{bwt}}
\def\cbwt{\texttt{bwt}}
\def\bwt{\texttt{ebwt}}
\def\cebwt{\texttt{c-ebwt}}
\def\lbcr{\texttt{bcr}}
\def\lcp{\texttt{lcp}}

\def\bwtS{\textsf{ebwt}(\mathcal{S})}

\def\bwtS#1{\textsf{ebwt}_{#1}(\mathcal{S})}
\def\lcpS#1{\textsf{lcp}_{#1}(\mathcal{S})}
\def\gsaS#1{\textsf{gsa}_{#1}(\mathcal{S})}
\def\bwtw#1{\textsf{bwt}_{#1}({w})}
\def\lcpw#1{\textsf{LCP}_{#1}({w})}
\def\rank{\textsf{rank}}
\def\select{\textsf{select}}
\def\smallW{{w}}
\def\bigS{\mathcal{S}}
\def\BCR{\texttt{BCR}}
\def\BWT{BWT}
\def\EBWT{EBWT}
\def\LCP{LCP}
\def\SA{SA}
\def\GSA{GSA}
\def\BCRLCP{\texttt{extLCP}}
\def\EGSA{\texttt{eGSA}}
\def\EMBWT{\texttt{BCRext} }
\def\EMBWTpp{\texttt{BCRext++} }
\def\BWTE{\texttt{bwte}}
\def\sort#1{\textrm{sort}(#1)}
\def\MEMLIM{\texttt{MEMLIMIT}}

\newcommand{\fbwt}{{\mathcal B}{\mathcal W}{\mathcal T}}

\section{Introduction}

The \emph{suffix array} (\SA), the \emph{longest common prefix array} (\LCP) and the \emph{Burrows-Wheeler transform} (\BWT) are data structures with many important applications in stringology \cite{ohlebusch2013book}.

The \LCP\  array of a string contains the lengths of the longest common prefixes of the suffixes pointed to by adjacent elements of the suffix array of the string \cite{ManbeMyersr:1993}.
The most immediate utility of the \LCP\ is to speed up suffix array algorithms and to simulate the more powerful, but more resource consuming, suffix tree.
Indeed, the \LCP\  array, combined with \SA\ or \BWT, simplifies the algorithms for some applications such as the rapid search for maximal exact matches, shortest unique substrings and shortest absent words \cite{OhlebuschGogKugelSpire2010,BellerGogOhlebuschSchnattinger2013,Herold_Kurtz_Gieg2008,Abouelhoda2004}.

Real world situations may present us with datasets that are not a single string but a large collection of strings, examples being corpora of web pages or the data coming from ``next-generation'' DNA sequencing (NGS) technologies. It is common for the latter in particular to generate aggregations of hundreds of millions of DNA strings  (sometimes called ``reads''), leading to a need for algorithms that deal not only with the collective nature of the data but also with its large overall size. In this paper we give an algorithm for the construction of \LCP\ array that addresses both these requirements.

The importance of the \LCP\ array has meant that its construction has been well-studied in the literature. For instance, there are algorithms that work in linear time and in internal memory (cf. \cite{Kasai:2001}). Other algorithms work in semi-external memory (see for instance \cite{KarkkainenManziniPuglisi:2009}) or directly via \BWT\ (see \cite{BellerGogOhlebuschSchnattinger2013}).
In \cite{Karkkainen:2006}, the first algorithm establishing that \LCP\ array can
be computed in $O(sort(n))$ complexity in the external memory model (i.e., the complexity of
sorting $n$ integers in external memory) is presented. Recently, in \cite{Bingmann13} an external memory algorithm to construct the suffix array of a string based on the induced sorting principle is provided. Such an algorithm can be augmented to also construct the \LCP\ array. The overhead in time and I/O operations for this extended algorithm over plain suffix array construction is roughly two.
In another recent paper \cite{KarkkainenKempa2014}, the authors introduced an external memory \LCP\ array construction algorithm for a string.
The main idea in \cite{KarkkainenKempa2014} is to divide the string into blocks that are small enough to fit in RAM and then scan the rest of the string once for each block. Such a strategy needs $16n$ bytes of disk space, where $n$ is the length of the string.

One's initial strategy for building the \LCP\ array of a collection might therefore be to concatenate its members into a single string and then apply one of the above methods. However, some stumbling blocks become apparent, and it is not immediate how these single-string approaches can be adapted to circumvent them. First, many existing algorithms for computing the \LCP\ array require data structures of size proportional to the input data to be held in RAM, which has made it impractical to compute the \LCP\ array of a very large string.
Moreover, by definition, the values in \LCP\ array should not exceed the lengths of the strings, so one would ideally employ distinct end-marker symbols to act as separators between the strings of the collection. However, assigning a different end-marker to each string is not feasible when the number of strings in the collection is very large. On the other hand, the use of the same end-marker symbol throughout the collection could give rise to \LCP\ values with the undesirable properties of sometimes exceeding the lengths of the strings and having dependencies on the order in which the strings are concatenated.

In the literature, the problem of the computation of \LCP\ array for a collection of strings has been considered in \cite{Shi:1996} and in \cite{LouzaTellesCiferri2013}. Moreover, a preliminary version of the results presented in this paper is in \cite{BauerCoxRosoneSciortino2012}. In particular, this paper includes full proofs of theorems from \cite{BauerCoxRosoneSciortino2012} and more detailed examples.
Furthermore, additional experimental results are also described. Finally, in this paper we introduce a simple solution for dealing with strings having different lengths, allowing \BCRLCP\ to work on any collection of strings.

Defining $N$ and $K$ as the sum of the lengths of all strings and the length of the longest string in the collection respectively, the approach in \cite{Shi:1996} requires $O(N\log K)$ time, but the $O(N \log N)$ bits of internal memory needed to store the collection and its \SA\ in internal memory prevents the method from scaling to massive data.

The external memory algorithm (called \EGSA) in \cite{LouzaTellesCiferri2013} builds both suffix and \LCP\ arrays for a collection of strings.
Such interesting strategy has an overhead in working disk space, so it seems that \EGSA\ cannot be used for huge collections of strings.

The methodology presented in this paper attempts to overcome the above mentioned limitations.
In fact, the aim of our paper is to design a lightweight algorithm for the computation, at the same time, of the longest common prefix array and the Burrows-Wheeler transform of a very large collection of strings having different or same length.
The computation is realized by performing disk data accesses only via sequential scans.
The total disk space usage never needs more than twice the output size, excluding the disk space required for the input.

In our approach, we directly compute both data structures, without needing to concatenate the strings of the collection and without requiring pre-computed auxiliary information such as the suffix array of the collection.

In order to do this, our algorithm is built upon the approach introduced in \cite{BauerCoxRosoneTCS2013} related to an extension of the \emph{Burrows-Wheeler transform} to a collection of strings defined in \cite{MantaciRRS07}.

In particular, our algorithm (called \BCRLCP) adds to the strategy in \cite{BauerCoxRosoneTCS2013} (called \BCR) some lightweight data structures, and allows the simultaneous computation of both the longest common prefix and the Burrows-Wheeler transform of a collection of $m$ strings. Such a computation is performed in $O((m+\sigma^2) \log N)$ bits of memory, with a worst-case time complexity of $O(K(N+\sort{m}))$, where $\sort{m}$ is the time taken to sort $m$ integers, $\sigma$ is the size of the alphabet, $N$ is the sum of the lengths of all strings and $K$ is the length of the longest string. Note that \BCRLCP\ needs at most $(2N-m)(\log \sigma + \log K)+N\log \sigma$ bits of disk space and requires $O\left(NK / (B\min(\log_\sigma N,\log_K N)) \right)$ disk I/O operations, where $B$ is the disk block size.

The low memory requirement enables our algorithm to scale to the size of dataset encountered in human whole genome sequencing datasets: in our experiments, we compute the \BWT\ and \LCP\ of collections as large as $800$ million strings having length $100$.

Moreover, \BCRLCP\ allows to compute also the suffix array of the strings of a collection (called \emph{generalized suffix array}).
Such a further computation increases the number of I/O operations, but it does not need further data structures in internal memory.

Section~\ref{sec:prel} gives preliminaries that we will use throughout the paper, in Section~\ref{sec:ADS_collection} we define the main data structures for a collection of strings, Section~\ref{bcrStrategy} describes an incremental strategy for computing the \BWT\ of a collection of strings of any length. Section~\ref{sec:algorithm} shows the sequential computation of the \LCP\ array for the collection. We present details on the efficient implementation of the algorithm and its complexity in Sections~\ref{sec:implementation} and \ref{sec:BCRLCP_complexity}, respectively. Computational results on real data are described in Section \ref{sec:experiments}. Section \ref{sec:conclusion} is devoted to some conclusions.

\section{Preliminaries}
\label{sec:prel}

Let $\Sigma =\{c_1, c_2, \ldots, c_\sigma\}$ be a finite ordered alphabet with $c_1< c_2< \ldots < c_\sigma$, where $<$ denotes the standard lexicographic order.
We append to a finite string $w\in \Sigma^*$ an end-marker symbol $\$$ that satisfies $\$ < c_1$. We denote its characters by $w[1], w[2],\ldots,w[k]$, where $k$ is the \emph{length} of $w$, denoted by $|w|$.
Note that, for $1 \leq i \leq k-1$, $w[i]\in \Sigma$ and $w[k]=\$ \notin \Sigma$.
A \emph{substring} of a string $w$ is written as $w[i,j] = w[i] \cdots w[j]$, with a substring $w[1,j]$ being called a \emph{prefix}, while a substring $w[i,k]$ is referred to as a \emph{suffix}.
A range is delimited by a square bracket if the correspondent endpoint is included, whereas the parenthesis means that the endpoint of the range is excluded.

We denote by $\bigS=\{w_0,w_1,\ldots,w_{m-1}\}$ the collection of $m$ strings of length at most $K$.
We suppose that to each string $w_i$ is appended an end-marker symbol $\$_i$ smaller than $c_1$, and $\$_i<\$_j$ if $i<j$. Let us denote by $N$ the sum of the lengths of all strings in $\bigS$.

For $j=0,\ldots,|w_i|-1$, we refer to the suffix $w_i[|w_i|-j,|w_i|]$ of a string $w_i$ as its \emph{$j$-suffix}; the \emph{0-suffix} of $w_i$ contains $\$_i$ alone.
The length of a $j$-suffix is equal to $j$ (up to considering the end-marker).
Let us denote by $\bigS_j$ the collection of the $j$-suffixes of all the strings of $\bigS$.

In our algorithm presented in this paper we use a unique end-marker $\$=c_0$ for all strings in $\bigS$, because we set $w_s[|w_s|] < w_t[|w_t|]$ if and only if $s < t$, so that if two strings $w_s$ and $w_t$ share the $j$-suffix, then $w_s[|w_s|-j,|w_s|] < w_t[|w_t|-j,|w_t|]$ if and only if $s < t$.
However, to ease the presentation distinct end-markers are shown.

We say that the symbol $w_i[|w_i|-j-1]$ is \emph{associated with} the $j$-suffix of $w_i$ for $j=0,\ldots,|w_i|-2$, because $w_i[|w_i|-j-1]$ precedes the $j$-suffix of $w_i$, i.e. the suffix $w_i[|w_i|-j,|w_i|]$. Moreover, we assume that $w_i[|w_i|]=\$_i$ is associated with the $(|w_i|-1)$-suffix of $w_i$, i.e. $w_i[1,|w_i|]$.

\section{SA, LCP and BWT of a collection of strings}\label{sec:ADS_collection}
Suffix array, longest common prefix array and Burrows-Wheeler transform are all commonly defined with reference to a single string. This section describes the extension of such notions to a collection of strings.

The \emph{suffix array} \SA\ of a string $w$ is an array containing the permutation of the integers $1,2, \ldots, |w|$ that arranges the starting positions of the suffixes of $w$ into lexicographical order.
There exist some natural extensions of the suffix array to a collection of strings (see \cite{Shi:1996}).

We define the \emph{generalized suffix array} $GSA$ of the collection $\bigS=\{w_0,w_1,\ldots,w_{m-1}\}$ as the array of $N$ pairs of integers $(t,j)$, corresponding to the lexicographically sorted suffixes $w_{j}[t, |w_{j}|]$, where $1 \leq t \leq |w_{j}|$ and $0 \leq j \leq m-1 $. In particular, $GSA[q]=(t,j)$ is the pair corresponding to the $q$-th smallest suffix of the strings in $\bigS$, i.e. to the suffix $w_{j}[t, |w_{j}|]$.

The \emph{longest common prefix array} of a string contains the lengths of the longest common prefixes of the suffixes pointed to by adjacent elements of \SA\ of the string \cite{PuglisiTurpin2008}.
The \emph{longest common prefix array} \LCP\ of a collection $\bigS$ of strings, denoted by $\lcp(\bigS)$, is an array storing the length of the longest common prefixes between two consecutive suffixes of $\bigS$ in the lexicographic order. For every $j=1, \ldots, N-1$, if $GSA[j-1]=(p_1,p_2)$ and $GSA[j]=(q_1,q_2)$, $LCP[j]$ is the length of the longest common prefix of suffixes starting at positions $p_1$ and $q_1$ of the words $w_{p_2}$ and $w_{q_2}$, respectively. We set $LCP[1]=0$.

For $i < j$, a \emph{range minimum query} $RMQ(i, j)$ on the interval $[i, j]$ in the \LCP\ array returns an index $k$ such that $LCP[k] = \min \{LCP[l]: i \leq l \leq j\}$.
If $GSA[r]=(p_1,p_2)$ and $GSA[s]=(q_1,q_2)$, it is not difficult to show that the length of the longest common prefix between the suffixes starting at positions $p_1$ and $q_1$ of the words $w_{p_2}$ and $w_{q_2}$ corresponds to $LCP[RMQ(r + 1, s)]$.

The suffix array of a string is related to the \emph{Burrows-Wheeler transform} introduced in \cite{bwt94}.
The original Burrows and Wheeler transform (\BWT) on a string is described as follows:
given a word $w\in \Sigma^*$, the output of \BWT\ is the pair $(L,I)$ obtained by lexicographically sorting the list of the conjugates of $w$. In particular, $L$ is the word obtained by concatenating the last symbol of each conjugate in the sorted list and $I$ is the position of $w$ in such a list. For instance, if $w=mathematics$ then $(L,I)=(mmihttsecaa,7)$\footnote{Note that, in the original \BWT, the symbol $\$$ is not appended to the input string $w$.}.

Actually, in several implementations of the \BWT, in order to improve the efficiency of the computation, one can consider a variant of the \BWT\ by the sorting the suffixes of $w\$$ rather than the conjugates of $w$.
To ensure the reversibility of the transform, one needs to append the symbol $\$$ at the end of the input string $w =w[1]\cdots w[k-1]$, where $w[i] \in \Sigma$, $\$ \notin \Sigma$ and $\$ < a \in \Sigma$. Hence the $\$$ symbol is at the position $k$ of $w$, so that $w =w[1]\cdots w[k-1] w[k]$. In this variant, the output $\lbwt(w)$ is a permutation of $w$, obtained as concatenation of the letters that (circularly) precede the first symbol of the suffix in the lexicographically sorted list of its suffixes: for $i=1, \ldots, k$, $\lbwt(w)[i]=w[SA[i]-1]$; when $SA[i]=1$, then $\lbwt(w)[i]=\$$ (it wraps around). In other words, the $i$-th symbol of the \BWT\ is the symbol just before the $i$-th suffix. For instance, if $w=mathematics\$$ then $\lbwt(w)=smmihtt\$ecaa$.
Note that, in this case, the second output of the \BWT\ (the index $I$) is not useful, because one can use the position of the $\$$-symbol for recovering the input string.

The Burrows-Wheeler transform can be extended to a collection of strings.
In its original definition \cite{MantaciRRS07} (see also \cite{MantaciRRS08}), such a reversible transformation (called \EBWT) produces a string that is a permutation of the characters of all strings in $\bigS$ and it does not make use of any end-marker.
The \EBWT\ of a collection $\bigS$ is a word (denoted by $\cebwt(\bigS)$)) obtained by letter permutation of the words in $\bigS$ together a set of indexes (denoted by ${\mathcal I}$) used to recover the original collection.
In particular, $\cebwt(\bigS)$ is obtained by concatenating the last symbol of each element in the sorted list of the conjugates of the words in $\bigS$. The sorting exploits an order relation defined by using lexicographic order between infinite words. For instance, if $\bigS=\{abac, cbab, bca, cba\}$, the output of \EBWT\  of $\bigS$ is the couple $(ccbbbcacaaabba, \{1,9,13,14\})$.

In this paper we use a more efficient variant of \EBWT\  of a collection of strings, that needs to append a different end-marker to each string of $\bigS$. In this case the word obtained as output, denoted by $\bwt(\bigS)$, is obtained by concatenating the symbols just preceding each suffix of the list of the suffixes of the words in $\bigS$ in according with the lexicographic order.
The output $\bwt(\bigS)$ can be also defined in terms of the generalized suffix array of $\bigS$. In particular, if $GSA[i]=(t,j)$ then $\bwt(\bigS)[i] = w_j[(t-1)]$; when $t=1$, then $\bwt(\bigS)[i]=\$_j$.
By using the above example, $\bwt(\bigS)=cbaacbb\$_0bacca\$_2ab\$_3\$_1$.
The external memory methods for computing the output $\bwt(\bigS)$ are given in \cite{BauerCoxRosoneCPM11}.
Furthermore, in practice, such methods use a unique end-marker rather than $m$ different end-markers.
Note that $\bwt(\bigS)$ differs, for at least $m$ symbols, from $\lbwt$ applied to the string obtained by concatenating all strings in $\bigS$.

\section{Computing the \EBWT\ of a collection of strings having any length}\label{bcrStrategy}
Our approach for computing the \LCP\ array is built upon the \BCR\ algorithm introduced in \cite{BauerCoxRosoneTCS2013} to compute  $\bwt(\bigS)$, where $\bigS$ is a collection of strings.
Note that \BCR\ in \cite{BauerCoxRosoneTCS2013} is described for collections of strings of fixed length. In this section we focus on the description of \BCR\ on collections of strings of any length.

The \BCR\ algorithm \cite{BauerCoxRosoneTCS2013} computes the $\bwt(\bigS)$ of the collection of strings $\bigS = \{w_0,w_1,\ldots,w_{m-1}\}$ without concatenating the strings. In the sequel, we assume that $0 \leq i \leq m-1$.
Note that we consider the symbol $\$_i$ appended to each string $w_i$.
We suppose that $K$ is the maximal length (including the end-markers) of the strings in $\bigS$ and $N$ the sum of their lengths (including the end-markers). We assume that $K \ll m$.

In Subsection \ref{subsec:BCR_algo} we describe how the algorithm works on a collection of strings of any length, Subsection \ref{subsec:BCR_DS} is devoted to detail the involved data structures needed to allow a lightweight implementation and to reduce the I/O operations.

\subsection{The algorithm}\label{subsec:BCR_algo}
The basic idea of the \BCR\ algorithm is to scan all the strings $w_0,w_1,\ldots,w_{m-1}$ in the collection $\bigS$ from right to left at the same time. This means that, at each iteration, it considers a ``slice'' of (at most) $m$ characters from the strings in the collection. \BCR\ builds incrementally, via $K$ iterations, the Burrows-Wheeler transform of $\bigS$ by simulating, step by step, the insertion of all suffixes having the same length in the list of sorted suffixes.

At the end of each iteration $j=0,\ldots, K-1$, \BCR\ algorithm builds a partial $\bwt(\bigS)$ (denoted by $\bwtS{j}$). It is referred to as partial, because if one inserts the $\$$-characters in their correct position (rather than the symbols that precede the $j$-suffixes), then one immediately obtains the $\bwt(\bigS_j)$ of all $j$-suffixes of $\bigS$  (if the length of some string is greater than $j$).
For instance, if $\bigS=\{AATACACTGTACCAAC\$_0, GAACAGAAAGCTC\$_1\}$ (including the distinct end-markers), $\bwtS{3}$ corresponds to the  computation of $\bwt(\{AAC\$_0, CTC\$_1\})$ when the end-markers are inserted.

The key point of this strategy is to establish the positions where the new symbols associated with the $j$-suffixes must be inserted in $\bwtS{j-1}$ in order to obtain $\bwtS{j}$. In other words, we have to find the position of each $j$-suffix in the list of sorted $t$-suffixes with $t < j$, without explicitly computing such a list.

In particular, at the step $j=0$, we have to consider the symbols associated with the $0$-suffixes and establish how they must be concatenated in order to obtain $\bwtS{0}$. Since we use (implicit) distinct end-markers and $\$_i  < \$_j$ if $i<j$, then it is easy to verify that $\bwtS{0}$ is obtained by considering the last symbol of each string $w_i$ and by concatenating them in the same order as the strings appear in the collection: $\bwtS{0}=w_0[|w_0| - 1]w_1[|w_1| - 1] \cdots w_{m-1}[|w_{m-1}| - 1]$.

At the iteration $1 \leq j < K-1$, we need to retain $\bwtS{j-1}$, keep track of the positions within it of the symbols $w_i[|w_i|-j]$ associated with the $(j-1)$-suffixes of $\bigS$ and build $\bwtS{j}$ by inserting at most $m$ new symbols into $\bwtS{j-1}$, i.e. $w_i[|w_i|-j-1]$ for $j < |w_i|-1$ or $\$$ for $j = |w_i|-1$.
Such operations simulate the insertion of the $j$-suffixes into the list of sorted suffixes computed in the previous step.
Note that if $j \geq |w_i|$ then the string $w_i$ is not considered.

The process ends at the step $K-1$, when \BCR\ considers the $j$-suffixes of the strings $w_i$ with $|w_i|=K$ and  inserts the end-markers of such strings into $\bwtS{K-2}$, in order to obtain $\bwtS{K-1}$, i.e. the $\bwt(\bigS)$.

In order to find, at each step $j \geq 1$, the positions where the new (at most) $m$ symbols must be inserted into $\bwtS{j-1}$, we use the notions related to backward search, such as the table $C$ and the $\rank$ function that have been extensively used in FM-index (see \cite{Ferragina:2000}). Formally, given a string $w$, $C_w[x]$ is a table of $\sigma$ integers that, for each character $x \in \Sigma$, contains the number of occurrences of lexicographically smaller characters in the string $w$. Moreover, given a symbol $x$, an integer $r$ and a string $w$, the function $\rank(x, r, w)$ returns the number of occurrences of character $x$ in the prefix $w[1, r]$.

In our context we suppose that, at the step $j-1 \geq 0$, we have inserted the symbol $x=w_i[|w_i|-j]$ at the position $r$ of $\bwtS{j-1}$.
At the step $j \geq 1$, for each $i$ with $j < |w_i|$, the position $p$, where we have to insert the new symbol associated with the $j$-suffix $w_i[|w_i|-j, |w_i|]$
into $\bwtS{j-1}$, is computed in the following way:
$$p=C_{\bwtS{j-1}}[x] + \rank(x, r, \bwtS{j-1}) + 1.$$

Note that such a formula corresponds to the computation of the so-called $LF$-mapping, a fundamental operation of the FM-index.

\subsection{Data Structures} \label{subsec:BCR_DS}

In order to decrease the I/O operations, \BCR\ does not use the table $C$ and computes the \rank\ function on segments of $\bwtS{j-1}$  rather than on the entire $\bwtS{j-1}$. More precisely, \BCR\ considers  $\bwtS{j}$ as the concatenation of $\sigma +1$ segments $B_j(0),B_j(1), \ldots,B_j(\sigma)$, where the symbols in $B_j(0)$ are the characters preceding the lexicographically sorted $0$-suffixes of $\bigS_j$ (such suffixes consist in only the end-marker symbols) and the symbols in $B_j(h)$, with $h \geq 1$, are the characters preceding the lexicographically sorted suffixes of $\bigS_j$ starting with $c_h \in \Sigma$. It is easy to see that, $\bwtS{0}=B_0(0)$ and the segments $B_0(h)$ (for $h=1,\ldots,\sigma$) does not contain any symbols.

Now, we can omit the table $C$, because it is only useful for determining the segment where the suffix must be inserted.
Indeed, the table $C$ simply allows us to establish that if the $j$-suffix starts with the symbol $c_z= w_i[|w_i| - j]$ then it is larger than the suffixes starting with a symbol smaller than $c_z$. This is equivalent to say that the $j$-suffix must be inserted in the block containing the suffixes starting with the symbol $c_z$, i.e. in $B_{j}(z)$.

 One can verify that the position $r$ in $\bwtS{j-1}$ corresponds to a position $s$ in $B_{j-1}(v)$ where $c_v$ is the first symbol of $(j-1)$-suffix of $w_i$, i.e $c_v=w_i[|w_i| - (j-1)]$.

Now, the new symbol $w_i[|w_i|-j-1]$ (or the end-marker $\$$ for the last step) must be inserted in the position $r$ of $B_{j}(z)$, where $r$ is obtained by computing the number of occurrences of $c_z=w_i[|w_i|-j]$ in $B_{j-1}(0),\ldots, B_{j-1}(v-1)$ and in $B_{j-1}(v)[1,s]$.  Moreover, the computation of $B_{j-1}(v)[1,s]$ corresponds to the computation of the function $\rank(c_z, s, B_{j-1}(v))$.
Note that, in order to compute the occurrences in $B_{j-1}(0),\ldots, B_{j-1}(v-1)$, a table of $O(\sigma^2 \log (N))$ bits of memory can be used.

We remark that for each step $j$, our strategy computes, for each $j$-suffix $\tau$, its definitive position in the sorted list of the elements of $\bigS_j$,  regardless of the fact that some $j$-suffixes lexicographically smaller than $\tau$ have not been considered yet in the current step $j$. Actually, this means that the algorithm computes the definitive position (called \emph{absolute position}) of the symbol associated to $\tau$ in the correspondent segment $B_j$.

During each iteration $j$, we use the arrays $U$, $N$, $P$, $Q$ containing at most $m$ elements each.  The values in such arrays are updated during each step. For ease of presentation we denote by $U_j$, $N_j$, $P_j$, $Q_j$ the arrays at the step $j$, described as follows.
For each $q = 0, \ldots , m-1$:
\begin{itemize}
\item $U_{j}$ is an array that stores all the symbols, if they exist, located at the position $j+1$ from the right. More formally, $U_{j}[q]=w_q[|w_q|-j-1]$ if $j<|w_q|-1$, $U_{j}[q]=\$_q$ if $j=|w_q|-1$, $U_{j}[q]=\#$ if $j>|w_q|-1$. Note that $\#$ is a symbol that does not appear in any string of the collection and it is not involved in the computation of $\bwtS{j}$.
    The array $U_{j}$ takes $O(m \log \sigma)$ bits of workspace.
\item $N_{j}$ is an array of integers such that $N_{j}[q]=i$ if the $j$-suffix of the string $w_i \in \bigS$ (with $j<|w_i|$) is the $q$-th $j$-suffix in the lexicographic order. It uses $O(m \log m)$ bits of workspace.
\item $P_{j}$ is an array of integers such that $P_{j}[q]$ is the absolute position of the symbol $w_i[|w_i| - j - 1]$ (or the end-marker $\$_i$), associated with the $j$-suffix of $w_i$ (with $j<|w_i|$), in $B_{j}(z)$, where $i = N_{j}[q]$ and $c_z$ is the first symbol of the $j$-suffix of $w_i$, i.e. $c_z = w_i[|w_i| - j]$.
It needs $O(m \log N)$ bits of workspace.
\item $Q_{j}$ is an array of integers such that $Q_j[q]$ stores the index $z$ where $c_z=w_i[|w_i|-j]$ where $N_{j}[q]=i$, i.e. the first symbol of the $j$-suffix of $w_i$ (with $j<|w_i|$). It uses $O(m \log \sigma)$ bits of workspace.
\end{itemize}

Note that, at each step $j$ ranging from $0$ to $K-1$, the algorithm considers only the first $t$ values of arrays $P_j$, $Q_j$ and $N_j$ where $t$ is the number of the strings of the collection having length greater than or equal to $j$.

At the start of the iteration $j$, we compute the new values in $Q_{j}$, $P_{j}$, $N_j$  in the following way. We suppose that the $(j-1)$-suffix of $w_i$ is the $t$-th $(j-1)$-suffix in lexicographic order, i.e. $i = N_{j-1}[t]$. Hence, we know that the first symbol of the $j$-suffix of $w_i$ is $c_z = w_i[|w_i| - j]$ and has been inserted in the position $P_{j-1}[t]=s$ in $B_{j-1}(v)$, where $Q_{j-1}[t]=v$ and $c_v$ is the first symbol of $(j-1)$-suffix of $w_i$, i.e $c_v=w_i[|w_i|-(j-1)]$.

The position $r$ of the new symbol $w_i[|w_i|-j-1]$ (or the end-marker $\$_i$ for $j=|w_i|$) that must be inserted in $B_{j}(z)$ is obtained by computing the number of occurrences of $c_z=w_i[|w_i|-j]$ in $B_{j-1}(0),\ldots, B_{j-1}(v-1)$ and in $B_{j-1}(v)[1,s]$, where $s=P_{j-1}[t]$ is the position of $c_z$ in $B_{j-1}(v)$.

Note that during the iteration $j$, we store the values in $P_j$, $Q_j$, $N_j$ and $U_j$ by replacing the corresponding values in $P_{j-1}$, $Q_{j-1}$, $N_{j-1}$ and $U_{j-1}$, so $Q_j[t]=z$, $P_j[t]=r$ and $N_j[t]=i$.

Finally, we sort $Q_{j}$, $P_{j}$, $N_j$ where the first and the second keys of the sorting are the values in $Q_j$ and $P_j$, respectively. So that we can insert, in sequential way, the symbols $U_j[q]$ (for $q = 0, \ldots, m-1$), into each segment $B_{j}(h)$ for $h=0,\ldots,\sigma$.

\begin{example}\label{ex:ebwt}
We suppose that $\bigS=\{w_0,w_1\}=\{AATACACTGTACCAAC\$_0, GAACAGAAAGCTC\$_1\}$ (including the distinct end-markers).
The algorithm starts, at the step $j=0$, by computing the $\bwtS{0}$. The $0$-suffixes are $\$_0,\$_1$, and the new symbols that we have to insert in $B_0(0)$ (the segment associated with the suffixes starting with the end-marker) are $C$ and $C$. We set $U_0=[C,C]$, $Q_0=[0,0]$ and $N_0=[0,1]$.
Since, we use (implicit) distinct end-markers and $\$_0 < \$_1$, we set $P_0=[1,2]$. So, $\bwtS{0}=B_0(0)=CC$.

Then, we consider $U_1=[A,T]$. Both symbols in $U_1$ should be inserted into $B_1(2)$ because both the associated suffixes ($C\$_0$ and $C\$_1$) start with the symbol $C$. So, we set $Q_1=[2,2]$ and $N_1=[0,1]$. The position in $B_1(2)$ of the symbol $A$ associated with the $1$-suffix of $w_0$ is equal to $1$ in the segment $B_1(2)$, indeed $C\$_0$ is the smallest suffix in $\bigS_{1}$ starting with the letter $C$.  The position in $B_1(2)$ of the symbol $T$ associated with the $1$-suffix of $w_1$ is equal to $2$ in the segment $B_1(2)$, because the symbol $C$ associated with $\$_1$ follows the symbol $C$ associated with $\$_0$. So we set $P_1=[1,2]$ and obtain $B_1(2)$ by inserting $A$ and $T$. Then $\bwtS{1}=CCAT$.

During the third step, the array $U_2=[A,C]$ is considered. Since the last inserted symbol of $w_0$ is $A$ then the new symbol $A$ associated with $AC\$_0$ must be inserted in the segment $B_2(1)$, whereas since the last inserted symbol of $w_1$ is $T$ then the new symbol $C$ associated with $TC\$_1$ must be inserted in the segment $B_2(4)$. So we set $Q_2=[1,4]$, $N_2=[0,1]$.
Since the number of occurrences of $A$ in $B_{1}(0),B_{1}(1)$ and in $B_{1}(2)[1,1]$ is $1$ then we have to insert the symbol $A$ at the position $1$ in $B_2(1)$.
Since the number of occurrences of $T$ in $B_{1}(0),B_{1}(1)$ and in $B_{1}(2)[1,2]$ is $1$ then we have to insert the symbol $C$ at the position $1$ in $B_2(4)$.
Since $AC\$_0$ and $TC\$_1$ are the smallest suffixes starting with $A$ and $T$ respectively, we set $P_2=[1,1]$.

The first three iterations are depicted in Figure \ref{fig:ThreeIterations}. The process continues via left extensions of suffixes until all symbols have been inserted.

\begin{figure}[!htb]
{\scriptsize
$$
\begin{array}{|c|c|}
\hline
\multicolumn{2}{|c|}{\mbox{Iteration 0}} \\
\hline
\bwtS{0}   & \mbox{Suffixes of $\bigS_0$} \\
\hline
B_{0}(0) & \\
\hline
\textbf{C}     & \$_0 \\
\hline
\textbf{C}     & \$_1 \\
\hline
	B_{0}(1) & \\
\hline
      &   \\
\hline
B_{0}(2) & \\
\hline
      &   \\
\hline
B_{0}(3) & \\
      &   \\
\hline			
B_{0}(4) & \\
      &   \\
\hline		
      &   \\
\hline		
\end{array}%
\qquad
\begin{array}{|c|c|}
\hline
\multicolumn{2}{|c|}{\mbox{Iteration 1}} \\
\hline
\bwtS{1}   & \mbox{Suffixes of $\bigS_1$} \\
\hline
B_{1}(0) & \\
\hline
C     & \$_0 \\
\hline
C     & \$_1 \\
\hline
B_{1}(1) & \\
\hline
&   \\
\hline
B_{1}(2) & \\
\hline
\textbf{A}     & C\$_0 \\
\hline
\textbf{T}     & C\$_1 \\
\hline
B_{1}(3) & \\
\hline
    &  \\
\hline
B_{1}(4) & \\
\hline		
      &   \\
\hline			
\end{array}
\qquad
\begin{array}{|c|c|}
\hline
\multicolumn{2}{|c|}{\mbox{Iteration 2}} \\
\hline
\bwtS{2}   & \mbox{Suffixes of $\bigS_2$} \\
\hline
B_{2}(0) & \\
\hline
C     & \$_0 \\
\hline
C     & \$_1 \\
\hline
B_{2}(1) & \\
\hline
\textbf{A}     & AC\$_0 \\
\hline
B_{2}(2) & \\
\hline
A     & C\$_0 \\
\hline
T     & C\$_1 \\
\hline
B_{2}(3) & \\
\hline
      &   \\
\hline		
B_{2}(4) & \\
\hline
\textbf{C}     & TC\$_1 \\
\hline
\end{array}
$$
}
\caption{$\bigS=\{w_0,w_1\}=\{AATACACTGTACCAAC\$_0, GAACAGAAAGCTC\$_1\}$. Note that the list of sorted suffixes is inserted in the figure in order to make it easier to understand.
In each table, the first column represents $\bwtS{j}=B_j(0)B_j(1)B_j(2)B_j(3)B_j(4)$ after each iteration. The new symbols inserted in each iteration are shown in bold.
It is verify that at the end of iteration $0$, we have $U_0=[C,C]$, $P_0=[1,2]$, $Q_0=[0,0]$, $N_0=[0,1]$. At the end of iteration $1$, we have $U_1=[A,T]$, $P_1=[1,2]$, $Q_1=[2,2]$, $N_1=[0,1]$. At the end of iteration $2$, we have $U_2=[A,C]$, $P_2=[1,1]$, $Q_2=[1,4]$, $N_2=[0,1]$.}\label{fig:ThreeIterations}
\end{figure}

\end{example}

\section{\LCP\ computation of a collection of strings via \EBWT}
\label{sec:algorithm}

The main goal of this section consists in the description of the strategy for computing, by using the \EBWT, the \LCP\ array of a massive collection of strings via sequential scans of the disk data. In particular, the main theorem of the section enables the simultaneous computation of both \LCP\ and \EBWT\ of a string collection $\bigS=\{w_0,w_1,\ldots,w_{m-1}\}$ of maximum length $K$. We recall that the last symbol of each string $w_i$ is the (implicit) end-marker $\$_i$.

Our method follows the \BCR\ algorithm in the sense that it scans all the strings from right-to-left in $K$ steps and simulates, at the step $j$, the insertion of suffixes of length $j$ into the sorted list of suffixes. This time, however, we wish to compute both the \LCP\ and \EBWT. So, at the step $j$, the longest common prefix array (denoted by $\lcpS{j}$) of the collection $\bigS_j$ of the suffixes having length at most $j$ is computed alongside $\bwtS{j}$. As well as computing an \LCP\ value for the inserted suffix, we must also modify the \LCP\ value for the suffix that comes after it to reflect the longest prefix common to it and the inserted suffix.

Our goal in this section is to frame these calculations in terms of $RMQ$s on sets of intervals within $\lcpS{j}$, serving as a preliminary to the next section, where we show how these computations may be arranged to proceed via sequential scans of the data.

It is easy to see that when $j=K-1$, $\lcpS{j}$ coincides with $\lcp(\bigS)$, the \LCP\ array of $\bigS$.
Since all $m$ end-markers are considered distinct, the longest common prefix of any pair of $0$-suffixes is $0$, so the first $m$ positions into $\lcpS{j}$ are $0$ for any $j\geq 0$.

The general idea behind our method is described in the following and depicted in Figure \ref{fig:ideaLCP}. At each step $j>0$,  the value of a generic element $\lcpS{j}[r]$ (with $r>m$) have to be computed by taking into account the $j$-suffix $\tau$ that is placed at the position $r$ of the sorted suffixes of $\bigS_j$. Such a value depends on the $f$-suffix $\tau_1$ placed at the position $r-1$ in the sorted list and, moreover, could lead an updating of the value $\lcpS{j}[r+1]$ (if it exists) corresponding to the $g$-suffix $\tau_2$, as shown in Figure \ref{fig:ideaLCP}. By using our method, the computation of $\lcpS{j}[r]$ can be realized by using the values of the arrays $\lcpS{j-1}$ and $\bwtS{j-1}$. In fact, if $\tau=c\gamma$ (where $c$ is a symbol and $\gamma$ is the $(j-1)$-suffix placed at a certain position $t$ in the sorted list of the suffixes of $\bigS_{j-1}$), then $c=\bwtS{j-1}[t]$.
Let us denote $\tau_1=a\gamma_1$ and $\tau_2=b\gamma_2$, where $a$ and $b$ are symbols and $\gamma_1$ and $\gamma_2$ are, respectively, the $(f-1)$-suffix and the $(g-1)$-suffix of $\bigS_{j-1}$ placed at the positions $d_1$ and $d_2$.  One can see that $\lcpS{j}[r]$ is equal to $0$ if $c\neq a$, otherwise it is equal to $1$ plus the longest common prefix between $\gamma_1$ and $\gamma$ computed in the array $\lcpS{j-1}$. Moreover $\lcpS{j}[r+1]$ is equal to $0$ if $c\neq b$, otherwise it is equal to $1$ plus the longest common prefix between $\gamma$ and $\gamma_2$ computed in the array $\lcpS{j-1}$. In this section we show how such a computation can be sequentially performed.

\begin{figure}[!htb]
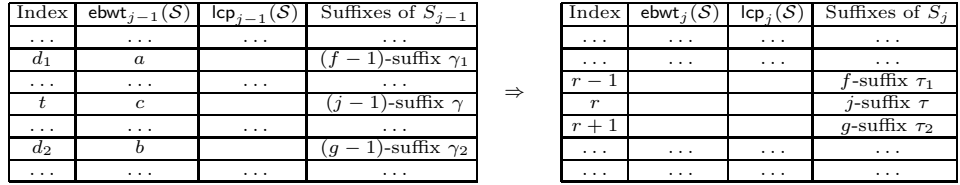

	{\scriptsize
		$$
		\begin{array}{|@{\ }c@{\ }|c@{\ }|c@{\ }|c@{\ }|}
		\hline
	\mbox{Index} & \bwtS{j-1} & \lcpS{j-1} & \mbox{Suffixes of } S_{j-1} \\
	\hline
	\ldots     & \ldots     & \ldots     & \ldots \\
	\hline
	d_1   & a     &       & (f-1)\mbox{-suffix } \gamma_1\\
	\hline
	\ldots     & \ldots     & \ldots     & \ldots \\
	\hline
	t     & c     &       & (j-1)\mbox{-suffix } \gamma\\
	\hline
	\ldots     & \ldots     & \ldots     & \ldots \\
	\hline
	d_2   & b     &       & (g-1)\mbox{-suffix } \gamma_2\\
	\hline
	\ldots     & \ldots     & \ldots     & \ldots \\
	\hline
\end{array}
\qquad
\begin{array}{@{\ }c@{\ }}
\hspace{-0.3cm}\Rightarrow\hspace{-0.5cm}
\end{array}
\qquad
\begin{array}{|@{\ }c@{\ }|c@{\ }|c@{\ }|c@{\ }|}
\hline
	\mbox{Index} & \bwtS{j} & \lcpS{j} & \mbox{Suffixes of } S_{j} \\
\hline
\ldots     & \ldots     & \ldots     & \ldots \\
\hline
\ldots     & \ldots     & \ldots     & \ldots \\
\hline
r-1   &       &       & f\mbox{-suffix } \tau_1\\
\hline
r     &       &       & j\mbox{-suffix } \tau\\
\hline
r+1   &       &       & g\mbox{-suffix } \tau_2\\
\hline
\ldots     & \ldots     & \ldots     & \ldots \\
\hline
\ldots     & \ldots     & \ldots     & \ldots \\
\hline
\end{array}%
$$
}
\caption{Iteration $j-1$ (on the left) and iteration $j$ (on the right).  }\label{fig:ideaLCP}
\end{figure}

Note that $\lcpS{j}$ can be considered the concatenation of $\sigma +1$ arrays $L_j(0),L_j(1), \ldots,L_j(\sigma)$ where, for $h=1,\ldots,\sigma$, the array $L_j(h)$ contains the values corresponding to the lengths of the longest common prefix of the suffixes of $\bigS_j$ that start with  $c_h \in \Sigma$, while $L_j(0)$ (corresponding to the $0$-suffixes) is an array of $m$ zeroes.
It is easy to see that $\lcpS{0}=L_0(0)$ and that $L_0(h)$ is empty for $h>0$.
We note that, for each $1 \leq h \leq \sigma$, $L_j(h)[1]=0$ and $L_j(h)[i]\geq 1$ for $i>1$, because the suffixes associated with such values share at the least the symbol $c_h \in \Sigma$.

As shown in Section \ref{bcrStrategy}, $\bwtS{j}$ can be partitioned in an analogous way into segments $B_j(0),B_j(1), \ldots,B_j(\sigma)$.
Given the segments $B_j(h)$ and $L_j(h)$, $h=0,\ldots,\sigma$, for the symbol $x$ occurring at position $r$ of $B_j(h)$ we define the $(j,h)$\emph{-LCP Current Interval} of $x$ in $r$ (denoted by $LCI_j^h(x,r)$) as the range $(d_1,r]$ in $L_{j}(h)$ (so we set $LCI_j^h(x,r)=L_j(h)(d_1,r]$), where
$d_1$ is the greatest position smaller than $r$ of the symbol $x$ in $B_{j}(h)$, if such a position exists.
If such a position does not exist, we define $LCI_j^h(x,r)=L_j(h)[r]$.
Analogously, we define for the symbol $x$ the $(j,h)$\emph{-LCP Successive Interval} of $x$ in $r$ (denoted by $LSI_j^h(x,r)$) as the range $(r,d_2]$ in $L_{j}(h)$ (so we set $LSI_j^h(x,r)=L_j(h)(r, d_2]$), where $d_2$ is the smallest position greater than $r$ of the symbol $x$ in $B_{j}(h)$, if it exists.
If such a position does not exist we define $LSI_j^h(x,r)=L_j(h)[r]$.

In order to compute the values $d_1$ and $d_2$ we use the function $\select$ that together with $\rank$ function play a fundamental role in FM-index. In particular, $\select(r,c,L)$ takes in input a symbol $c$, an integer $r$ and a string $L$ and finds the position of the $r$-th occurrence of $c$ in $L$.

In our context, if $d_1$ and $d_2$ exist, then $B_{j}(h)[t] \neq x$ for $t=d_1+1, \ldots, r-1$ and for $t=r+1, \ldots, d_2-1$, so it is easy to verify that $d_1=\select(\rank(x,r,B_{j}(h))-1,x,B_{j}(h))$ and $d_2=\select(\rank(x,r,B_{j}(h))+1,x,B_{j}(h))$.

We observe that the computation of the minimum value into $LCI_j^h(x,r)$ and $LSI_j^h(x,r)$ is equivalent to the computation of $RMQ(d_1, r)$ and $RMQ(r, d_2)$, respectively. We cannot directly compute these values, because we build each $B_j(h)$ and $L_j(h)$ in sequential way, so we do not know the left extreme of $LCI_j^h(x,r)$ and the right extreme of $LSI_j^h(x,r)$ of all symbols that we have to insert in each $B_j(h)$, $h=1,\ldots,\sigma$.

The following theorem, related to \cite[Lemma 4.1]{GogOhlebuschO11}, shows how to compute the segments $L_j(h)$, with $j > 0$, by using $L_{j-1}(h)$ and $B_{j-1}(h)$ for any $h > 0$.

\begin{theorem}\label{th:LCP_case_allsuffix}
Let $\mathcal{I}=\{r_0 < r_1< \ldots< r_{l-1}\}$ be the set of the positions where the symbols associated with the $j$-suffixes starting with the letter $c_z$ must be inserted into $B_j(z)$.
For each position $r_b \in \mathcal{I}$ ($0 \leq b < l$),
$$L_j(z)[r_b]=\left\{\begin{array}{ll}
                     0 & \mbox{ if $r_b=1$} \\
                     1 & \mbox{ if $r_b>1$ and $LCI_{j-1}^v(c_z,t)=L_{j-1}(v)[t]$} \\
                     \min LCI_{j-1}^v(c_z,t) +1& \mbox{otherwise}
                   \end{array}\right.
                   $$
where $c_v$ is the first character of the $(j-1)$-suffix of $w_i$, and $t$ is the position in $B_{j-1}(v)$ of symbol $c_z$ preceding the $(j-1)$-suffix of $w_i$.\\

For each position $(r_b+1) \notin \mathcal{I}$ (where $r_b \in \mathcal{I}$ and $0 \leq b < l$), then
$$L_j(z)[r_b+1]=\left\{\begin{array}{ll}
                     1 & \mbox{ if $LSI_{j-1}^v(c_z,t)=L_{j-1}(v)[t]$} \\
                     \min LSI_{j-1}^v(c_z,t) +1& \mbox{otherwise}
                   \end{array}\right.
                   $$

For each position $s$, where $1 \leq s<r_b$ (for $b=0$), $r_{b-1} < s < r_b$ (for $0 < b < l-1$), $s> r_b$ (for $b=l-1$) then
$$L_j(z)[s]=L_j(z)[s-b]$$
\end{theorem}

\begin{proof}
We consider a generic position $r \in \mathcal{I}$ corresponding to the position where the new symbol $w_i[|w_i|-j-1]$ (or $w_i[|w_i|]=\$_i$) must be inserted into $B_{j}(z)$ and the corresponding value must be inserted into $L_{j}(z)$.
The new symbol $w_i[|w_i|-j-1]$ precedes the $j$-suffix $w_i[|w_i|-j,|w_i|]$. Such a suffix is obtained by concatenating the symbol $c_z=w_i[|w_i|-j]$ with its $(j-1)$-suffix.
Let us suppose that the symbol $c_z$ associated with $(j-1)$-suffix starting with the symbol $c_v$ is in position $t$ in $B_{j-1}(v)$.
Hence, $t$ is also the position of the $(j-1)$-suffix in the lexicographic order among the suffixes of $\bigS$ of length at most $(j-1)$ starting with $c_v$.

In order to find  $L_j(z)[r]$, we have to distinguish two cases.

In the first case, the $j$-suffix is the smallest suffix (in the lexicographic order) of length at most $j$ starting with the symbol $c_z$, i.e. $r=1$. This means that $c_z$ does not exist in any segment $B_{j-1}(h)$, $h=0,\ldots,v-1$ and there does not exist in $B_{j-1}(v)[1,t-1]$. In this case $LCI_{j-1}^v(c_z,t)=L_{j-1}(v)[t]$ and $L_j(z)[1]=0$.

In the second case, there exists some suffix starting with $c_z$ of length at most $j$ that is lexicographically smaller than the $j$-suffix, i.e. $r>1$.
Recall that $L_j(z)[r]$ represents the length of the longest common prefix between the $j$-suffix and a $f$-suffix of a string $w_{q}$ (for some $0 \leq q \leq m-1$) , with $f\leq j$, starting with the symbol $c_z$, that immediately precedes the $j$-suffix in the lexicographic order.

If the longest common prefix between the $f$-suffix of $w_q$ and the $j$-suffix of $w_i$ is $c_z$ then
$LCI_{j-1}^v(c_z,t)=L_{j-1}(v)[t]$ and $r>1$, so $L_j(z)[r]=1$. This means that the symbol $c_z$ associated to the $(f-1)$-suffix of $w_{q}$ is not contained in the segment $B_{j-1}(v)$, but it is contained in some $B_{j-1}(h)$, $h=0,\ldots,v-1$.

If the longest common prefix between the $f$-suffix of $w_q$ and the $j$-suffix of $w_i$ is longer than $c_z$
then both the $(f-1)$-suffix of $w_{q}$ and the $(j-1)$-suffix of $w_i$ start with the same symbol $c_v$. So, we can suppose that the symbol associated with the $(f-1)$-suffix of $w_q$ is at the position $d_1$ in the segment $B_{j-1}(v)$. Remark that symbols in position $d_1$ and $t$ are equal to $c_z$.
Then $LCI_{j-1}^v(c_z,t)=L_{j-1}(v)(d_1,t]$ and $L_j(z)[r] = \min (LCI_{j-1}^v(c_z,t)) +1$.

Recall that $L_j(z)[r+1]$ represents the length of the longest common prefix between the $j$-suffix and a $g$-suffix of a string $w_{p}$ (for some $0 \leq p \leq m-1$), with $g\leq j$, starting with the symbol $c_z$, that immediately follows the $j$-suffix in the lexicographic order.

If the longest common prefix between the $j$-suffix of $w_i$ and $g$-suffix of $w_p$ is $c_z$ then
$LSI_{j-1}^v(c_z,t)=L_{j-1}(v)[t]$, so $L_j(z)[r+1]=1$. This means that the symbol $c_z$ associated to the $(g-1)$-suffix of $w_{p}$ is not contained in the segment $B_{j-1}(v)$ (but it is contained in some $B_{j-1}(h)$, $h=v+1,\ldots,\sigma$).

If the longest common prefix between the $j$-suffix of $w_i$ and the $g$-suffix of $w_p$ is longer than $c_z$
then both the $(j-1)$-suffix of $w_i$ and the $(g-1)$-suffix of $w_{p}$ start with the same symbol $c_v$.  So, we can suppose that the symbol associated with the $(g-1)$-suffix of $w_p$ is at the position $d_2$ in the segment  $B_{j-1}(v)$. Remark that symbols in position $r$ and $d_2$ are equal to $c_z$.
Then $LSI_{j-1}^v(c_z,t)=L_{j-1}(v)(t, d_2]$ and $L_j(z)[r+1] = \min(LSI_{j-1}^v(c_z,t)) +1$.

Note that the position $r+1$ in $L_{j}(z)$ does not exist when the the $j$-suffix is the greatest suffix (in the lexicographic order) of length at most $j$ starting with the symbol $c_z$. The suffix that, eventually, immediately follows the $j$-suffix in lexicographic order is involved in a segment $B_{j}(h)$, for some $h=z+1,\ldots, \sigma$, hence such suffix starts with a symbol greater than $c_z$.
\end{proof}

A consequence of the theorem is that the segments $B_{j}$ and $L_{j}$ can be constructed sequentially and stored in external files.
This fact will be used in the next section.

\section{Lightweight implementation via sequential scans}\label{sec:implementation}

Based on the strategy described in the previous section, here we propose an algorithm (named \BCRLCP) that simultaneously computes the \EBWT\ and the \LCP\ of a collection of strings $\bigS$.
Memory use is minimized by reading data sequentially from files held in external memory: only a small proportion of the symbols of $\bigS$ need to be held in internal memory.
We could also add the computation of the generalized suffix array of $\bigS$ without adding further data structures.

As in the previous sections, we assume that the collection $\bigS$ comprises $m$ strings of length at most $K$, that $j = 0,1,\ldots,K-1$, $i=0,1,\ldots,m-1$, $h=0,1,\ldots \sigma$, that $\bwtS{j}=B_j(0)B_j(1) \cdots B_j(\sigma)$ and $\lcpS{j}=L_j(0)L_j(1) \cdots L_j(\sigma)$.
When $j=K-1$,  $\bwtS{j}=\bwt(\bigS)$ and $\lcpS{j}=\lcp(\bigS)$. We also assume that $q=1,\ldots,m$.

Note that, at each iteration $j$, both the segments $B_{j}$ and $L_{j}$, initially empty, are stored in different external files that replace the files used in the previous iteration.
Consequently, both $\bwtS{j}$ and $\lcpS{j}$ are updated accordingly.

\subsection{Basic idea of the algorithm}\label{subsec:BCRLCP_algo}

The main part of the algorithm \BCRLCP\ consists of $K$ consecutive iterations. At iteration $j$, we consider all the $j$-suffixes of $\bigS$ and simulate their insertion in the sorted suffixes list.
For each symbol associated with the  $j$-suffix that we have to insert at the position $r$ into $B_j(h)$, we also have to insert the new values related to the longest common prefix at position $r$ and $r+1$ into $L_j(h)$, where $c_h$ is the first symbol of the considered $j$-suffix.

In order to compute $\bwtS{j}$ and $\lcpS{j}$, the algorithm  \BCRLCP\  needs  to hold six arrays of $m$ integers in internal memory. Four of these ($P$, $Q$, $N$ and $U$) are as employed by the algorithm \BCR\ (see Section \ref{bcrStrategy}) and further two arrays ($C$ and $S$) are needed to compute and update the values of the longest common prefixes at each iteration.
As for $P$, $Q$, $N$ and $U$ arrays (see Section \ref{bcrStrategy}), for ease of presentation we denote by $C_j$ and $S_j$ the arrays at the $j$-th iteration. They contain exactly one integer for each string, i.e. they use $O(m \log K)$ bits of workspace. They are sequentially computed by using other auxiliary data structures described in Subsection \ref{subsec:BCRLCP_DS}.

More formally, if $i=N_j[q]$ then $C_j[q]$ stores the length of the longest common prefix between the $j$-suffix of $w_i$ and the previous suffix (if it exists) in the list of sorted suffixes with respect to the lexicographic order of all the suffixes of $\bigS$ of length at most $j$, whereas $S_j[q]$ contains the length of the longest common prefix between the $j$-suffix of $w_i$ and the successive suffix $\tau$ in the list of sorted suffixes (if it exists). Such values will be computed at the iteration $j-1$ according to Theorem \ref{th:LCP_case_allsuffix}. Note that $S_{j}[q]$ is used when the suffix $\tau$ exists and $P_{j}[q]+1 \neq P_{j}[q+1]$.

Example \ref{{exLCP}} provides an idea of the running of the algorithm \BCRLCP\ and shows how the involved arrays are used.

\begin{example}\label{{exLCP}}
Figure \ref{fig:compute} illustrates an execution of our method on the collection $\bigS=\{w_0,w_1\}=\{AATACACTGTACCAAC\$_0, GAACAGAAAGCTC\$_1\}$ at two consecutive iterations. Note that we have appended different end-marker to each string ($\$_0$ and $\$_1$, respectively).  In particular, we suppose that at the iteration $j=12$, we have computed the arrays $P_{12}$, $Q_{12}$ and $N_{12}$.
We recall that we have computed $C_{12}$ and $S_{12}$ in the previous iteration. At the end of the first phase of the iteration $12$ we have $N_{12}=[1,0]$, $U_{12}=[G,A]$, $P_{12}=[3,4]$,  $Q_{12}=[1,2]$, $C_{12}=[3,2]$, $S_{12}=[2,2]$. Note that they are sorted by using the first and the second keys the values in $Q_{12}$ and $P_{12}$ respectively.
In this way the new $B_{12}(h)$-segments (on the left in the figure) have been constructed by adding the bold characters.
We can compute $LCI_{12}^1(G,3)$ and $LSI_{12}^1(G,3)$ and $LCI_{12}^2(A,4)$ and $LSI_{12}^2(A,4)$ useful to terminate the iteration $12$.
We obtain that $LCI_{12}^1(G,3)$ corresponds to the range $(1,3]$ in $L_{12}(1)$.
So the minimum value is $2$ and the value in $\lcpS{13}$ associated with the $13$-suffix $\tau$ (i.e. $\tau=GAACAGAAAGCTC\$_1$) of $w_1$ is $3$ (this value is stored into $C_{13}$).
Now we have to compute the value of the longest common prefix between the suffix $\tau$ and the suffix that immediately follows $\tau$ in the lexicographic order and to store this value into $S_{13}$.
Since the symbol $G$ does not appear in the range $L_{12}(1)(3,10]$, it means that there are not suffixes starting with $GA$ lexicographically greater than $\tau$, so such value will be less than $2$.
Because the symbol $G$ does appear at least once in $B_{12}(h)$ with $h>1$ (in this case $B_{12}(2)[7]=G$), it means that there exists at least a suffix starting with $G$ lexicographically greater than $\tau$ (in this case the suffix $GCTC\$_1$).
So the value in $\lcpS{j+1}$ of the suffix that follows $\tau$ must be updated to $1$ (i.e. we store this value in $S_{13}$).
Similarly, since $LCI_{12}^2(A,4)=L_{12}(2)(1,4]$, the minimum value is $1$ and so the value in $\lcpS{13}$ for the $13$-suffix $\omega$ (i.e. $\omega=ACACTGTACCAAC\$_0$) of $w_0$ is $2$ (i.e. we insert this value in $C_{13}$).
Moreover, $LSI_{12}^2(A,4)=L_{12}(2)(4,5]$, the minimum value is $2$ and hence the value in $\lcpS{13}$ of the suffix that follows $\omega$ must be updated to $3$ (i.e. we insert this value in $S_{13}$).

At the iteration $j=13$, we compute the arrays $P_{13}$, $Q_{13}$ and $N_{13}$, whereas the arrays $C_{13}$ and $S_{13}$ have been computed in the previous iteration.
So we have $N_{13}=[0,1]$, $U_{13}=[T,\$_1]$, $P_{13}=[6,2]$,  $Q_{13}=[1,3]$, $C_{13}=[2,3]$, $S_{13}=[3,1]$.
While the new $B_{13}(h)$-segments (on the right in the figure) are being constructed, we can sequentially insert and update the new values in $\bwtS{j+1}$ and $\lcpS{j+1}$ and compute the new values $C_{14}$ and $S_{14}$.

\begin{figure}[!htb]
{\scriptsize
  $$
    \begin{array}{@{\ }c@{\ }c@{\ }c@{\ }l@{\ }}
  &  L_{12}(0) &  B_{12}(0) & \mbox{Suffixes of $\bigS_{12}$} \\
  & 0 & C &       \$_0 \\
  & 0 & C &       \$_1 \\
  &  L_{12}(1) &  B_{12}(1) & \mbox{Suffixes of $\bigS_{12}$} \\
\multicolumn{ 1}{c|}{} & 0 & G & AAAGCTC\$_1 \\
\cline{2-2}
\multicolumn{ 1}{c|}{LCI^1_{12}(G,3)} & 2 & C &    AAC\$_0 \\
\multicolumn{ 1}{c|}{\rightarrow} &    {\bf 3} &    {\bf G} & {\bf AACAGAAAGCTC\$_1} \\
  & 2 & A & AAGCTC\$_1 \\
  & 1 & A &     AC\$_0 \\
  & 2 & A & ACAGAAAGCTC\$_1 \\
  & 2 & T & ACCAAC\$_0 \\
  & 2 & C & ACTGTACCAAC\$_0 \\
  & 1 & C & AGAAAGCTC\$_1 \\
  & 2 & A &  AGCTC\$_1 \\
  &   &   &   \\
  &  L_{12}(2) &  B_{12}(2) & \mbox{Suffixes of $\bigS_{12}$} \\
\multicolumn{ 1}{c|}{} & 0 & A &      C\$_0 \\
\cline{2-2}
\multicolumn{ 1}{c|}{LCI^2_{12}(A,4)} & 1 & T &      C\$_1 \\
\multicolumn{ 1}{c|}{} & 1 & C &   CAAC\$_0 \\
\multicolumn{ 1}{c|}{\rightarrow} &    {\bf 2} &    {\bf A} & {\bf CACTGTACCAAC\$_0} \\
\cline{2-2}
\multicolumn{ 1}{c|}{LSI^2_{12}(A,4)} & 2 & A & CAGAAAGCTC\$_1 \\
  & 1 & A &  CCAAC\$_0 \\
  & 1 & G &    CTC\$_1 \\
  & 2 & A & CTGTACCAAC\$_0 \\
  &  L_{12}(3) &  B_{12}(3) & \mbox{Suffixes of $\bigS_{12}$} \\
  & 0 & A & GAAAGCTC\$_1 \\
  & 1 & A &   GCTC\$_1 \\
  & 1 & T & GTACCAAC\$_0 \\
  &   &   &   \\
  &  L_{12}(4) &  B_{12}(4) & \mbox{Suffixes of $\bigS_{12}$} \\
  & 0 & G & TACCAAC\$_0 \\
  & 1 & C &     TC\$_1 \\
  & 1 & C & TGTACCAAC\$_0 \\
    \end{array}%
\qquad
    \begin{array}{@{\ }c@{\ }c@{\ }c@{\ }l@{\ }}
    & L_{13}(0) &  B_{13}(0) & \mbox{Suffixes of $\bigS_{13}$} \\
    &     0 & C &       \$_0 \\
    &     0 & C &       \$_1 \\
 & L_{13}(1) &  B_{13}(1) & \mbox{Suffixes of $\bigS_{13}$} \\
   &      0 & G & AAAGCTC\$_1 \\
   &      2 & C &    AAC\$_0 \\
   &      3 & G & AACAGAAAGCTC\$_1 \\
   &      2 & A & AAGCTC\$_1 \\
   &      1 & A &     AC\$_0 \\
   \rightarrow & {\bf 2} & {\bf T} & {\bf ACACTGTACCAAC\$_0} \\
   &      \textbf{\underline{3}} & A & ACAGAAAGCTC\$_1 \\
   &      2 & T & ACCAAC\$_0 \\
   &      2 & C & ACTGTACCAAC\$_0 \\
   &      1 & C & AGAAAGCTC\$_1 \\
   &      2 & A &  AGCTC\$_1 \\
 & L_{13}(2) &  B_{13}(2) & \mbox{Suffixes of $\bigS_{13}$} \\
   &      0 & A &      C\$_0 \\
   &      1 & T &      C\$_1 \\
   &      1 & C &   CAAC\$_0 \\
   &      2 & G & CACTGTACCAAC\$_0 \\
   &      2 & A & CAGAAAGCTC\$_1 \\
   &      1 & A &  CCAAC\$_0 \\
   &      1 & G &    CTC\$_1 \\
   &      2 & A & CTGTACCAAC\$_0 \\
 & L_{13}(3) &  B_{13}(3) & \mbox{Suffixes of $\bigS_{13}$} \\
   &     0 & A & GAAAGCTC\$_1 \\
  \rightarrow & {\bf 3} & {\bf \$_1} & {\bf GAACAGAAAGCTC\$_1} \\
   &      \textbf{\underline{1}} & A &   GCTC\$_1 \\
   &      1 & T & GTACCAAC\$_0 \\
& L_{13}(4) &  B_{13}(4) & \mbox{Suffixes of $\bigS_{13}$} \\
   &        0 & G & TACCAAC\$_0 \\
   &      1 & C &     TC\$_1 \\
   &      1 & C & TGTACCAAC\$_0 \\
     \end{array}%
   $$
}\caption{
Iteration $j=12$ (on the left) and iteration $j=13$ (on the right) on the collection $\bigS=\{AATACACTGTACCAAC\$_0,GAACAGAAAGCTC\$_1\}$.
The first two columns represent the $\lcpS{j}$ and the $\bwtS{j}$ after the iterations. The positions of the new symbols corresponding to the $13$-suffixes (shown in bold on the right) are computed from the positions of the $12$-suffixes (in bold on the left), which were retained in the array $P$ after the iteration $12$. The new values in $\lcpS{13}$ (shown in bold on the right) are computed during the iteration $12$ and are contained in $C_{12}$. The updated values in $\lcpS{13}$ (shown in bold and underlined on the right) are computed during the iteration $12$ and are contained in $S_{13}$.
}\label{fig:compute}
\end{figure}
\end{example}

\subsection{Sequential implementation} \label{subsec:BCRLCP_DS}
In this subsection we show how, at the generic iteration $j$ of the algorithm \BCRLCP, $\lcpS{j}$ is sequentially computed and updated by using the data structures previously described.
At the first iteration, $C_{0}[q]=0$ and $S_{0}[q]=0$ for each $q=1,\ldots,m$, because the $m$ end-markers are considered distinct.
Moreover, the algorithm initializes the segments $B_0$ and $L_0$ in the following way: $B_0(0)=w_0[|w_0|-1]w_1[|w_1|-1]\cdots w_{m-1}[|w_{m-1}|-1]$ and  $L_{0}(0)[q]=0$, for each $q$. Consequently, the arrays are initialized by setting $N_{0}[q]=q-1$, $P_{0}[q]=q$, $Q_{0}[q]=0$, $C_{1}[q]=1$  and $S_{1}[q]=1$, for each $q$.

Each iteration $j>0$ can be divided into two consecutive phases.

During the first phase we only read the segments $B_{j-1}$ in order to compute the arrays $P_j$, $Q_j$ $N_j$ and $U_j$. Then we sort $Q_{j}$, $P_{j}$, $N_j$, $C_{j}$, $S_j$, where the first and the second keys of the sorting are the values in $Q_j$ and $P_j$ respectively.
We omit the description of the first phase, because it can be found in Section \ref{bcrStrategy}, so we focus on the second phase.

In the second phase, the segments $B_{j-1}$ and $L_{j-1}$ are read once sequentially both for the construction of new segments $B_{j}$ and $L_{j}$ and for the computation of the arrays $C_{j+1}$ and $S_{j+1}$, as they will be used in the next iteration.
Moreover, the computation of the segments $L_j$ is performed by using the arrays $C_{j}$ and $S_{j}$ constructed during the previous step.

Since the identical elements in $Q_j$ are consecutive, we open the pair files $B_{j-1}(h)$ and $L_{j-1}(h)$ (for $h=0,\ldots,\sigma$) at most once.
Each of these files can be sequentially read, because the positions in $P_j$ are sorted in according with $Q_j$.

In the sequel, we focus on a single segment $B_j(z)$, for $z=1,\ldots,\sigma$, by assuming that $Q_j[p]=z$ for each $l \leq p \leq l'$, with $l \geq 1$ and $l' \leq m$, i.e. we are considering the elements in $B_j(z)$ and $L_j(z)$ associated with the suffixes starting with $c_z$.

Since $P_{j}[l]<  \ldots < P_{j}[l']$, we can sequentially build $B_j(z)$ and $L_j(z)$ by copying the old values from $B_{j-1}(z)$ and $L_{j-1}(z)$ respectively, by inserting each new symbol $U_j[N_{j}[p]]$ into $B_j(z)[P_{j}[p]]$ and $C_j[p]$ into $L_j(z)[P_{j}[p]]$ and by updating the value $L_j(z)[P_{j}[p]+1]$ with the value $S_j[p]$, if the position $P_{j}[p]+1\neq P_{j}[p+1]$ exists.

The crucial point consists in computing, at the same time, each value $C_{j+1}[p]$ and $S_{j+1}[p]$ (required for the next iteration) related to the $(j+1)$-suffix of all string $w_i$, with $i=N_{j}[p]$ and $j<|w_i|$, without knowing the left extreme $d_1$ of $LCI_j^h(x,r)$ and the right extreme $d_2$ of $LSI_j^h(x,r)$, where $r=P_{j}[p]$ and $x=U_j[N_{j}[p]]$ for each $p$.
In the sequel, we say that the left extreme is the \emph{opening position} and the right extreme is the \emph{closing position} of some symbol $x$.

\begin{figure}
\begin{center}\small \texttt{
   \begin{minipage}[c]{.95\linewidth}
    \linesnumbered
    \begin{algorithm}[H]%
      \nocaptionofalgo
      initAdditionalArrays\;
      $q = 0$\;
      \While{$q < m$}{	
		 $p \leftarrow q$; $z \leftarrow Q_j[p]$;  $s \leftarrow 1$\;
		 \While {($p < m$) $\wedge$ ($Q_j[p] = z$)} {
            \While {$s < P_j[p]$} {
                copy $x$ from $B_{j-1}(z)$ into $B_{j}(z)$\;
                copy $xLcp$ from $L_{j-1}(z)$ into $L_{j}(z)$\;
                {\sc updateLCI($x$)};    {\sc updateLSI($x$)}\;
                s++\;
            }
            {\sc insertNewSymbol($U_j[N_j[p]]$)}\;
            {\sc updateLCI($U_j[N_j[p]]$)};    {\sc updateLSI($U_j[N_j[p]])$}\;
            s++\;
            \If{$P_j[p] +1 \neq P_j[p+1]$} {
                copy $x$ from $B_{j-1}(z)$ into $B_{j}(z)$\;
                read $xLcp$ from $L_{j-1}(z)$\;
                insert $S_j[p]$ in $L_{j}(z)$\;
                {\sc updateLCI($x$)}; {\sc updateLSI($x$)}\;
                s++\;
            }
            p++\;
        }
        \While{not eof ($B_{j-1}(z)$)} {
           copy $x$ from $B_{j-1}(z)$ into $B_{j}(z)$\;
           copy $xLcp$ from $L_{j-1}(z)$ into $L_{j}(z)$\;
           {\sc updateLSI($x$)}\;
        }
        \For{$x \in \Sigma$}{
            \If {$isMinLSIop[x]$} {
            	$S_{j+1}[minLSINSeq[x]] \leftarrow 1$\;
            }
       }
       $q \leftarrow p$\;
     }
     \caption{{\sc $(L_j,B_j)$-construction}}
    \end{algorithm}
  \end{minipage}
}
\end{center}
\caption{Construction of the segments $L_j$ and $B_j$ for $j>0$.}
\label{algo:CS_computation}
\end{figure}

A pseudo-code that realizes the second phase of a generic iteration $j>0$ of the algorithm \BCRLCP\ can be found in Figure \ref{algo:CS_computation} and it uses the following additional arrays of $\sigma$ elements, defined as follows:
\begin{itemize}
	\item {\sc isMinLCIop} and {\sc isMinLSIop} are arrays of flags, where {\sc isMinLCIop}$[v]$ and
	{\sc isMinLSIop}$[v]$ indicate the opening of the $LCI_{j}^z$ and $LSI_{j}^z$ associated with $c_v$, respectively.
	\item {\sc minLCI} and {\sc minLSI} are arrays of integers, where {\sc minLCI}$[v]$ and {\sc minLSI}$[v]$ store the minimum value among the values in $L_j(z)$ from the opening position of the intervals associated with $c_v$ and the current position $s$.
	\item {\sc minLSInseq} is an array of integers, where {\sc minLSInseq}$[v]$ contains the index $g$ of the position in which {\sc minLSI}$[v]$ will be stored in $S_{j+1}$. This is useful, because when we close the $LSI_{j}^z$ associated with the symbol $c_v$ at some position $s$, we have to store {\sc minLSI}$[v]$ in some position $g$ of $S_{j+1}$, such that $P_{j}[g]<s$.
\end{itemize}

The routines {\sc insertNewSymbol}, {\sc updateLCI} and {\sc updateLSI} are described in Figure \ref{algo:insertNewSymbol} and Figure \ref{algo:updateMin}.

\begin{figure}[!htb]
\begin{center}\small \texttt{
   \begin{minipage}[c]{.95\linewidth}
    \linesnumbered
        \begin{algorithm}[H]%
      \nocaptionofalgo
            append $x$ to $B_{j}(z)$\;
            \If{$P_j[p]=1$} {
            	append $0$ to $L_{j}(z)$\;
            }
            \Else {
            	append $C_j[p]$ to $L_{j}(z)$\;
            }
            \If {{\sc isMinLCIop}$[x]$} {
            	  updatemin {\sc minLCI}$[x]$\;
            		$C_{j+1}[p] \leftarrow$ {\sc minLCI}$[x]+1$\;
            }
            \Else{
            		$C_{j+1}[p] \leftarrow 1$\;
            }
            {\sc isMinLCIop}$[x] \leftarrow 0$\;
            {\sc isMinLSIop}$[x] \leftarrow 1$\;
            {\sc minLSINSeq}$[x] \leftarrow p$\;
     \caption{{\sc insertNewSymbol($x$)}}
    \end{algorithm}
  \end{minipage}
}
\end{center}
\caption{Routine for inserting new values in $B_j(z)$ and $L_j(z)$.}
\label{algo:insertNewSymbol}
\end{figure}

{\sc insertNewSymbol($x$)} routine inserts at position $s=P_{j}[p]$ in $B_j(z)$ each new symbol $x=c_v=U_j[N_{j}[p]]$.
By Theorem \ref{th:LCP_case_allsuffix}, it follows that  $L_j(z)[P_{j}[p]]=0$ if $P_{j}[p]=1$ or $L_j(z)[P_{j}[p]]=C_j[p]$ otherwise. Moreover, the position $s=P_j[p]$ is surely:
\begin{itemize}
    		\item the closing position of $LCI_{j}^z(c_v, P_j[p])$.
				If {\sc isMinLCIop}$[v]=0$, then $P_j[p]$ is the position of the first occurrence of $c_v$ in $B_{j}(z)$, hence $LCI_{j}^z(c_v, P_j[p])=L_j(z)[P_j[p]]$ and we set $C_{j+1}[p]=1$ according to Theorem \ref{th:LCP_case_allsuffix}.
				Otherwise, {\sc isMinLCIop}$[v]$ has been set to $1$ in some position $d_1<P_j[p]$, so $LCI_j^z(c_v, P_j[p])=L_j(z)(d_1,P_j[p]]$ and we set $C_{j+1}[p] = \min(LCI_{j}^z(c_v, P_j[p]))+1$. Such minimum value is stored into {\sc minLCI}$[v]$. Moreover, we set {\sc isMinLCIop}$[v] =0$.
	\item the opening position of $LSI_{j}^z(c_v, P_j[p])$. So, we set {\sc isMinLSIop}$[v]=1$, the value {\sc minLSI}$[v]$ is updated and {\sc minLSInseq}$[v] = p$.
We observe that if the position $P_j[p]$ is the last occurrence of $c_v$ in $B_{j}(z)$ (this fact is discovered when the end of the file is reached), it means that $LSI_{j}^z(c_v, P_j[p]) =L_j(z)[P_j[p]]$, i.e. we set $S_{j+1}[p]=1$. 				
\end{itemize}

Note that in order to compute the values $C_{j+1}[p]$ and $S_{j+1}[p]$ for each $p$, we do not need to know the exact opening position $d_1$ of $LCI_j^z(U_j[N_{j}[p]],P_{j}[p])$ and the exact closing position $d_2$ of $LSI_j^z(U_j[N_{j}[p]],P_{j}[p])$, but we only need to compute the minimum values in these ranges.

\begin{figure}[!htb]
\begin{center}\small \texttt{
   \begin{minipage}[c]{.95\linewidth}
    \linesnumbered
        \begin{algorithm}[H]%
      \nocaptionofalgo
      \For{each $\alpha \in \Sigma$} {
          \If {{\sc isMinLCIop}$[\alpha]$} {
            updatemin {\sc minLCI}$[\alpha]$\;
          }
      }
      {\sc isMinLCIop}$[x] \leftarrow 1$\;
      init {\sc minLCI}$[x]$\;
     \caption{{\sc updateLCI($x$)}}
    \end{algorithm}
  \end{minipage}
}
\end{center}
\begin{center}\small \texttt{
   \begin{minipage}[c]{.95\linewidth}
    \linesnumbered
    \begin{algorithm}[H]%
      \nocaptionofalgo
      \For{each $\alpha \in \Sigma$} {
          \If {\sc {isMinLSIop}$[\alpha]$} {
            updatemin {\sc minLSI}$[\alpha]$\;
          }
      }
      \If {{\sc isMinLSIop}$[x]$}{
         $S_{j+1}[${\sc minLSINSeq}$[x]]\leftarrow ${\sc minLSI}$[x]$
      }
      {\sc isMinLSIop}$[x]\leftarrow 0$\;
      init {\sc minLSI}$[x]$\;
      init {\sc minLSINSeq}$[x]$\;
     \caption{{\sc updateLSI($x$)}}
    \end{algorithm}
  \end{minipage}
}
\end{center}
\caption{Routines for computing the minimum values in $LCI_j^z$ and $LSI_j^z$.}
\label{algo:updateMin}
\end{figure}

When at any position $s$ in $B_j(z)$ we insert the symbol $c_v$ (from $U_j$ or from $B_{j-1}(z)$), then the position $s$ is assumed to be:
\begin{itemize}
	\item the opening position of $LCI_{j}^z(c_v, y)$, if another occurrence of $c_v$ will be inserted, as new symbol, at some next position $y=P_j[f]$, for some $f$ such that $P_j[f]>s$. So, we set {\sc isMinLCIop}$[v] =1$ and the value {\sc minLCI}$[v]$ is updated (see {\sc updateLCI($x$)} routine in Figure \ref{algo:updateMin}).
	\item the closing position of $LSI_{j}^z(c_v, P_j[g])$, if another occurrence of $c_v$ has been inserted, as new symbol, at some previous position $P_j[g]$, for some $g$ such that $P_j[g] < s$. In this case, $LSI_{j}^z(c_v, P_j[g])=L_j(z)(P_j[g], s]$ (because {\sc minLSInseq}$[v]=g$) and 	we set $S_{j+1}[g] = \min (LSI_{j}^z(c_v, P_j[g]))+1$ according to Theorem~\ref{th:LCP_case_allsuffix}. Such a minimum value is stored into {\sc minLSI}$[v]$. We set {\sc isMinLSIop}$[v] = 0$ (see {\sc updateLSI($x$)} routine in Figure \ref{algo:updateMin}).
\end{itemize}	

When $B_j(z)$ is entirely built, the closing position of some $LSI_{j}^z(c_v, P_j[g])$ for some $l \leq g \leq l'$ could remain not found. So, we could have some value in {\sc isMinLSIop} equal to $1$. For each $c_v$ such that  {\sc isMinLSIop}$[v]=1$, the last occurrence of $c_v$ appears at position $P_j[g]$ (we recall that we have set {\sc minLSInseq}$[v]$ to $g$). In this case, $LSI_{j}^z(c_v, P_j[g])=L_j(z)[P_j[g]]$ and we set $S_{j+1}[g] =1$ according to Theorem~\ref{th:LCP_case_allsuffix} (see the for loop in Figure \ref{algo:CS_computation}).

One can verify that these steps work in a sequential way. Moreover, one can deduce that, while the same segment is considered, for each symbol $c_v \in \Sigma$ at most one $LCI_{j}^z(c_v,P_j[f])$ for some $l \leq f \leq l'$ and at most one $LSI_{j}^z(c_v,P_j[g])$ for some $l \leq g \leq l'$ will have not their closing position.

\section{Complexity of the algorithm \BCRLCP} \label{sec:BCRLCP_complexity}

The complexity of \BCRLCP\ algorithm depends mainly on $K$, i.e. the length of the longest string in the collection, because the algorithm works in $K$ passes and for each step it needs to build $B_j(z)$ and $L_j(z)$ from $B_{j-1}(z)$ and $L_{j-1}(z)$ for each $z=0,\ldots, \sigma$.
At each step $j$, the used internal memory depends on the number of strings in $\bigS$ of length greater than $j$. Such a value is upper bounded by $m$ for each step.
Note that, at each step, the total size of $\sigma+1$ files containing the partial $\bwt(\bigS)$ is increased by at most $m$ symbols. Analogously the total size of $\sigma+1$ files containing the partial $\lcp(\bigS)$ is increased by at most $m$ values. So, the used disk space mainly depends on $m$ for each step.

Since our algorithm accesses disk data only by sequential scans, we analyze it counting the number of disk passes as in the standard external memory model (see \cite{Vitter:2006}).
We denote by $B$ the disk block size and we assume that both the RAM size and $B$ are measured in units of $\Theta(\log N)$-bit words.

From the size of the data structures and from the description of the phases of the \BCRLCP\ algorithm given in previous sections, we can state the following theorem.

\begin{theorem}
Given a collection $\bigS$ of $m$ strings over an alphabet of size $\sigma$ where $K$ is the maximal length (including the end-markers) and $N$ is the sum of their length, \emph{\BCRLCP} algorithm simultaneously compute the \EBWT\ and the \LCP\ array of $\bigS$ by using $O\left(NK / (B\min(\log_\sigma N,\log_K N)) \right)$ disk I/O operations and $O((m+\sigma^2) \log N)$ bits of RAM in $O(K(N+\sort{m})$ CPU time, where $\sort{m}$ is the time taken to sort $m$ integers in internal memory. Moreover, \emph{\BCRLCP} needs at most $(2N-m)(\log \sigma + \log K)+N\log \sigma$ bits of disk space.
\end{theorem}

\begin{proof}
 At each iteration $j$, the main data structures used by \BCRLCP\ are $U_j$ of size $m \log \sigma$ bits, $Q_j$ of size at most $m \log \sigma$ bits, $N_j$ of size at most $m \log m$ bits, $P_j$ of size at most $m \log N$ bits, $C_j$ and $S_j$ of size at most $2m \log K$ bits.
 Moreover, we need $\sigma^2 \log N$ bits for computing the number of occurrences of each symbol in each segment $B_j(h)$, for $h=1,\ldots,\sigma$. The additional arrays take $O(\sigma \log m)$ (assuming that $K \ll m$).
 So, the workspace is $O((m+\sigma^2) \log N)$ bits.
 As we sort $Q_{j}$, $P_{j}$, $N_j$, $C_{j}$, $S_j$ where the first and the second keys of the sorting are the values in $Q_j$ and $P_j$, we need $O(mj+\sort{m})$ CPU time, where $\sort{m}$ is the time taken to sort $m$ integers in internal memory.
 The total CPU time is $O(K(N+\sort{m})$.
  We observe that the input is split into $K$ files, so that each file contains $m$ symbols, one for each string. At each step, such symbols will read and store into $U_j$ array. The amount of I/O operations for handling the input file is $O(\frac{N}{B \log_\sigma N})$.
 During the first phase of each iteration $j>0$, we need to read  at most $m(j-1) (\log \sigma + \log K)$ bits for computing $Q_j$, $P_j$, $N_j$ arrays. In the second phase we need to read at most $m(j-1) (\log \sigma + \log K)$ bits from $B_{j-1}$ and $L_{j-1}$ segments in order to obtain $C_{j+1}$, $S_{j+1}$ arrays and $B_{j}$ and $L_{j}$ segments by writing $mj (\log \sigma + \log K)$ bits of disk space.
The disk I/O operations for obtaining the output of each iteration is $O\left(\frac{mj}{B} (\frac{1}{\log_\sigma N} + \frac{1}{\log_K N})\right)$.
The total number of I/O operations is $O\left(\frac{NK}{B\log_\sigma N} + \frac{NK}{B \log_K N} \right)$.
\end{proof}

The internal memory of \BCRLCP\ can be reduced significantly by observing that rather than storing $P$, $Q$, $N$ and $C$ to internal memory, they could reside on disk because these arrays are sequentially processed.
In this way, the algorithm needs to store in internal memory the array $S$, {\sc isMinLCIop}, {\sc isMinLSIop}, {\sc minLCI}, {\sc minLSI} and {\sc minLSInseq} that require random accesses.
Clearly in this case one should use an external memory sorting algorithm.

\begin{remark}\label{rm:GSA}
It is easy to verify that the algorithm \BCRLCP\ can be also adapted for computing the generalized suffix array of $\bigS$ together the \EBWT.
Such further computation increases the number of I/O operations, but one do not need of further data structures in internal memory.
Note that, for each $i$, $GSA[q]=(t,j)$ is the pair corresponding to the $q$-th smallest suffix of the strings in $\bigS$, i.e. to the suffix $w_{j}[t, |w_{j}|]$.
Since we know at each iteration of \BCRLCP\ the values $t$ and $j$, because they are store into $P_j$ and $N_j$ arrays respectively, it is enough to modify the \BCRLCP\ code by adding the instructions for storing such values in external memory.
\end{remark}

\section{Computational experiments and discussion}
\label{sec:experiments}

Our algorithm \BCRLCP\ represents the first lightweight method that simultaneously computes, via sequential scans, the \LCP\ array and the \EBWT\ of a vast collection of strings.

We developed an implementation of the algorithm described in Section~\ref{sec:implementation}, which is available upon request from the authors\footnote{A more optimized version of \BCRLCP\ is available on {https://github.com/BEETL/BEETL}, but it only allows to use datasets of strings of fixed length.}.

Our primary goal has been to analyze the additional overhead in runtime and memory consumption of simultaneously computing both \EBWT\ and \LCP\ via \BCRLCP\ compared with the cost of using \BCR\ (cf. \cite{BauerCoxRosoneTCS2013}) to compute only the \EBWT.
For this goal, we used a publicly available collection of human genome strings from the Sequence Read Archive
\cite{SRA11} at {ftp://ftp.sra.ebi.ac.uk/vol1/ERA015/ERA015743/srf/} and created subsets containing $100$, $200$ and $800$ million reads, each read being 100 bases in length on the alphabet $\{A, C, G, T, N\}$.
\begin{table}[!htb]
\begin{center}
    \begin{tabular}{p{1.5cm}rp{.2cm}p{2cm}p{2cm}p{2cm}p{2cm}}
     \hline \hline
     \textbf{instance } & \textbf{size} & & \textbf{program} & \textbf{wall clock} & \textbf{efficiency} & \textbf{memory}\\
                        \hline \hline
   0100M & $9.31$  & & \BCR\    & $1.05$            & $0.81$ & $1.32$\\   
        & $9.31$  & & \BCRLCP & $4.03$            & $0.83$ & $1.50$\\ 
\hline
   0200M & $18.62$  & & \BCR\ & $1.63$            & $0.58$ & $2.62$ \\ %
         & $18.62$  & & \BCRLCP & $4.28$             & $0.79$ & $2.99$ \\ %
\hline
\hline
   0800M & $74.51$ & & \BCR\ & $3.23$            & $0.43$ & $10.24$\\ 
   		   & $74.51$ & & \BCRLCP & $6.68$           &  $0.67$ &  $12.29$\\  
\hline
\end{tabular}
\end{center}
\caption{The input string collections were generated on an Illumina GAIIx sequencer, all strings are $100$ bases long.
Size is the input size in GiB, wall clock time---the amount of time that elapsed
from the start to the completion of the instance---is given as microseconds per
input base, and memory denotes the maximal amount of memory (in GiB) used
during execution. The efficiency column states the CPU efficiency values,
i.e. the proportion of time for which the CPU was occupied and not waiting
for I/O operations to finish, as taken from the output of the
\texttt{/usr/bin/time} command.
}
\label{tab:inputInstances}
\end{table}

Table~\ref{tab:inputInstances} shows the results for the instances that we created. We show an increasing in runtime since  \BCRLCP\ writes the values of \LCP\ after that the symbols in \EBWT\ are written, so it effectively increases the I/O operations.
So, a time optimization could be obtained if we read/write at the same time both the elements in \EBWT\ and \LCP\ by using two different disks.
The tests on $100$ and $200$ million reads instances were done on the same machine, having $16$GiB of memory and two quad-core Intel Xeon E5450 $3.0$GHz processors.
The collection of $800$ million reads was processed on a machine with $64$GiB of RAM and four quad-core Intel Xeon E7330 $2.4$GHz processors. On both machines, only a single core was used for the computation.


For the experiments in the following, we use single identical nodes in a cluster of Dell PowerEdge M600 servers, each equipped with $8$ Intel(R) Xeon(R) CPU X5460 $3.16$GHz processors $6144$KiB cache and $32$GiB of RAM, not used in exclusive mode.



As pointed out in Section~\ref{bcrStrategy}, \BCRLCP\ is not restricted to work on collections of strings of equal length.
To examine the behaviour of our algorithm on datasets with reads of varying length and number, we have created five different datasets of DNA strings from datasets in Table~\ref{tab:inputInstances}: $A_0$ is composed of $4$ million strings of $100$bp, $A_1-A_4$ have been obtained rearranging the characters in $A_0$ so that they form strings whose distribution of lengths is specified in Table~\ref{tab:varlength}.

\begin{table}[!htb]
{\footnotesize
	\begin{center}
		\begin{tabular}{c|ccccc}
			\hline \hline
			\textbf{strings}	&  &  & \textbf{strings} & & \\
			\textbf{length}	&  &  & \textbf{number} & & \\
			\hline \hline
			& $A_0/P_0$ & $A_1/P_1$ & $A_2/P_2$ & $A_3/P_3$ & $A_4/P_4$ \\
			\hline \hline
			$10$  & $-$ & $-$ & $-$  & $1,600,000$ & $-$ \\
			$50$  & $-$ & $-$ & $-$  & $1,200,000$ & $1,700,000$ \\
			$100$ & $4,000,000$ & $1,000,000$ & $153,846$  & $690,000$ & $2,150,000$\\
			$250$ & $-$ & $25,000$ & $153,846$  & $-$ & $-$ \\
			$500$ & $-$ & $200,000$ & $153,846$  & $510,000$ & $100,000$ \\
			$750$ & $-$ & $125,000$ & $153,846$  & $-$ & $-$ \\
			$1000$ & $-$ & $100,000$ & $153,846$  & $-$ & $50,000$ \\
			\hline \hline
			\textbf{tot. strings} & $4,000,000$ & $1,450,000$ & $769,230$ & $4,000,000$ & $4,000,000$\\
			\hline
			\textbf{tot. char.}	& $400,000,000$  & $400,000,000$ & $399,999,600$ & $400,000,000$ & $400,000,000$\\
			\hline
			\textbf{max length}	& $100$  & $1,000$ & $1,000$ & $500$ & $1,000$\\
			\hline
		\end{tabular}
	\end{center}
	\caption{Lengths distribution of the strings in the datasets $A_0$, $A_1$, $A_2$, $A_3$, $A_4$ containing DNA strings and in the datasets  $P_0$, $P_1$, $P_2$, $P_3$, $P_4$ containing protein strings.}
	\label{tab:varlength}
}
\end{table}

Datasets $A_0$, $A_3$ and $A_4$ have the same number of characters and the same number of strings but different string lengths distribution and rising maximum string length. As can be observed from the first three rows of Table~\ref{tab:varlenres}, a sensible increase in wall clock time is registered for greater maximum string length.
This is consistent with the theoretical results on complexity described in Section \ref{sec:BCRLCP_complexity} where we show that the total number of I/O operations depends on the maximal length of the strings.
As expected, a corresponding decrease in CPU efficiency values is registered, while greater use of memory is due to an increment of the datatype size in $C$ and $S$ arrays.
Datasets $A_4$, $A_1$ and $A_2$ have approximately the same number of characters and exactly the same maximal string length but decreasing number of strings.
The corresponding rows of Table~\ref{tab:varlenres} show a decrement in wall clock time and memory usage.
Furthermore an increment of CPU efficiency values is registered, due to the fact that the number of strings in the collection affects I/O operations more than effective \EBWT\ and \LCP\  computation.

\begin{table}[!htb]
\begin{center}
		\begin{tabular}{cccccc}
			\hline
			\multicolumn{6}{c}{\BCRLCP} \\
			\hline \hline
			\textbf{instance} &  \textbf{wall clock} & \textbf{efficiency} & \textbf{mem} & \textbf{aux disk} & \textbf{tot disk}\\
			 &  &  &  &   \textbf{space} &  \textbf{space}\\
			\hline \hline
			$A_0$ &  $4.68$ & $0.77$  & $50$ ($0.2$\%) &  $1.13$ &  $1.88$ \\ %
			$A_3$ & $40.37$ & $0.56$ & $85$ ($0.3$\%)  &  $3.76$ & $5.64$ \\ %
			$A_4$ &  $166.37$ & $0.34$ & $85$ ($0.3$\%) &  $5.62$ & $7.50$ \\ %
			$A_1$ &  $100.09$ & $0.47$ & $31$ ($0.1$\%) & $3.23$ & $5.09$ \\ %
			$A_2$ &  $73.38$ & $0.57$ & $17$ ($0.1$\%) &  $2.59$ & $4.45$ \\ %
				\hline
				$P_0$  & $10.62$ & $0.90$  & $50$ ($0.2$\%) &  $1.13$ & $1.88$ \\ %
				$P_3$  & $77.57$ & $0.77$ & $85$ ($0.3$\%)  &  $3.76$ & $5.64$ \\ %
				$P_4$  & $200.65$ & $0.76$ & $85$ ($0.3$\%) &  $5.62$ & $7.50$ \\ %
				$P_1$  & $162.67$ & $0.75$ & $31$ ($0.1$\%) &  $3.23$ & $5.09$ \\ %
				$P_2$  & $145.68$ & $0.75$ & $17$ ($0.1$\%) &  $2.59$ & $4.45$ \\ %
\hline
			\end{tabular}
		\end{center}
	\caption{The amount of time that elapsed from the start to the completion of the instance is given as microseconds per input base is the \textit{wall clock} time. The \textit{efficiency} column provides the CPU efficiency values, i.e. the proportion of time for which the CPU was occupied and not waiting for I/O operations to finish, as taken from the output of the \texttt{/usr/bin/time} command. The column \textit{mem} is the maximum value returned by the \texttt{RES} field of \texttt{top} command representing the amount of used internal memory in MiB. The column \textit{aux disk space} denotes the maximum amount of external memory in GiB used to accommodate all the working files, \textit{tot disk space} includes the output files too.}
	\label{tab:varlenres}
\end{table}



We have also tested the \BCRLCP\ performances on datasets containing strings over alphabets with more than five symbols.
For this purpose we have created,  from UniProtKB/TrEMBL dataset\footnote{\texttt{ftp://ftp.uniprot.org/ (release November 2015)}}, five different datasets (denoted by $P_0-P_4$) whose protein strings have the same lengths distribution of the strings in the datasets $A_0-A_4$ (see Table \ref{tab:varlength}).
The results of Table \ref{tab:varlenres} show that the same considerations made for the DNA string collections $A_0-A_4$ also hold for the datasets $P_0-P_4$. Moreover, we notice a general increase in CPU efficiency and wall clock time, while the used internal and external space remains unchanged. Altogether, one can verify that both the varying string length and the size of the alphabet could have a significant influence on the \BCRLCP\ performance.

Note that we have not compared our method with other lightweight approaches that compute the \LCP\ array of a single string (see for instance \cite{Bingmann13}), because their implementations cannot be promptly tested on string collections, nor on a generic string produced as output by \EBWT\ of a string collection. In general, an entirely like-for-like comparison would imply the concatenation of the strings of the collection by different end-markers. However, for our knowledge, the existing implementations do not support the many millions of distinct end-markers our test collections would require.
An alternative is to concatenate each of strings with the same end-marker. The use of the same end-marker without additional information leads to values in the \LCP\ exceeding the lengths of the strings and depending on the order in which the strings are concatenated.

At the best of our knowledge, the only lightweight tool capable of working on large string collections in lightweight memory is \EGSA\ \cite{LouzaTellesCiferri2013}.
The algorithm \EGSA\footnote{\texttt{https://github.com/felipelouza/egsa}} takes in input a string collection $\bigS$ and returns the generalized suffix and \LCP\ arrays of $\bigS$. Note that, as well as \BCRLCP, \EGSA\ uses (implicit) distinct end-markers.

For large collections \EGSA\ works in two phases \cite{LouzaCom2015,LouzaThesis2013}. In the first phase, the collection $\bigS$ is partitioned into different subsets of equal dimension and each subset is treated as a single string, that is the concatenation of the strings of the subset. Then \EGSA\ produces, in internal memory, the $SA$ and the $LCP$ for each concatenated string and writes them to external memory.
In the second phase, \EGSA\ merges the arrays previously computed to obtain suffix and \LCP\ arrays of the entire collection $\bigS$.
The number of created subsets is related to the amount of internal memory available for the computation, such a value can be set through the parameter \MEMLIM.

We point out that using subsets of equal dimension in the first phase implies that the performance of \EGSA\ does not deteriorate when a large collection with strings of varying length is considered.

Moreover, while \BCRLCP\ computes the \LCP\ array via the \EBWT\ and is parameter-free, \EGSA\ builds the \LCP\ array by producing the $GSA$ for string collections and needs that \MEMLIM\ (or the number of subsets) is set.
The overall approach adopted by \EGSA\ could be thought as ``orthogonal'' to that of \BCRLCP, in the sense that \BCRLCP\ builds the \LCP\ array, incrementally, by inserting at each step a
``slice'' of (at most) $m$ characters from the strings in the collection, while \EGSA\ proceeds by building separately (in internal memory) the \LCP\ array of each string (or subset) and then opportunely merging them.

\begin{table}[!tb]
	\begin{center}
		\begin{tabular}{ccccccc}
			\hline \hline
			\textbf{dataset} &   \textbf{size}  & \textbf{strings} & \textbf{strings} & \textbf{mean}  & \textbf{LCP} & \textbf{LCP}\\
			                 &                  & \textbf{number}  & \textbf{length} & \textbf{length}  & \textbf{mean} & \textbf{max}\\
			\hline \hline
			$D_1$ &  $0.11$ &  $1,202,532$ & $76-101$ & $95.53$ & $38$ & $101$ \\ %
			$D_2$ &  $2.47$ &   $26,520,387$ & $100$ & $100$ & $21$ & $100$ \\ %
			$D_3$ &  $4.88$ &   $52,415,147$ & $100$ & $100$ & $25$ & $100$ \\ %
			$D_4$ &  $4.45$ &   $47,323,731$ & $100$ & $100$ & $25$ & $100$ \\ %
			\hline
		\end{tabular}
	\end{center}
	\caption{The column \textit{size} is the input size in GiB; the columns \textit{strings length} and \textit{mean length} correspond to the range of lengths and the average length in the dataset respectively; the columns \textit{LCP mean} and \textit{LCP max} are the average and maximum values of the longest common prefix of the strings in the dataset.
	}
	\label{tab:datasets}
\end{table}

\begin{table}[!htb]
{\scriptsize
	\begin{center}
		\begin{tabular}{ccccccccc}
			\hline \hline
			\textbf{instance} & \textbf{program} & \MEMLIM & \textbf{wall clock} & \textbf{efficiency} & \textbf{\% mem} &  \textbf{aux disk} & \textbf{tot disk} \\
			&  &  &  &  & & \textbf{space} & \textbf{space}\\
			\hline \hline
$D_1$ &  \BCRLCP              &    $-$         & $4.33$ & $0.72$     & $0.1$ &  $0.33$ & $0.55$ \\ %
      &  \EGSA                       &  $15$ MiB  & $2.36$ & $0.77$  & $0.1$ &  $3.80$ & $5.21$ \\
      &  \EGSA                       &  $5$ MiB    & $2.60$ & $0.82$ & $0.1$ &  $3.80$ & $5.21$ \\
   \hline
   $D_2$ & \BCRLCP                    & $-$        &  $5.07$ & $0.73$ & $0.9$ & $7.48$  & $12.47$ \\ %
   &  \EGSA                           & $300$ MiB  &  $2.95$ & $0.66$ & $0.9$ & $86.07$ & $118.50$ \\
   &  \EGSA                           &  $120$ MiB &  $2.99$ & $0.65$ & $0.4$ & $86.07$ & $118.50$ \\
   \hline   	
$D_3$ &  \BCRLCP              &    $-$                 & $4.31$ & $0.86$  & $2.5$ &  $14.79$ & $24.65$  \\ %
      &  \EGSA                       &  $800$ MiB      & $4.52$ & $0.55$  & $2.5$ &  $170.16$ & $234.25$ \\
      &  \EGSA                       &  $236$ MiB      & $4.95$ & $0.54$  & $0.7$ &  $170.16$ & $234.25$ \\	
\hline
$D_4$ &  \BCRLCP                     &    $-$          & $9.33$ & $0.91$  & $2.2$ & $13.35$  & $22.26$ \\ %
      &  \EGSA                       &  $723$ MiB      & $4.19$ & $0.63$  & $2.2$ & $155.11$ & $212.98$ \\
      &  \EGSA                       &  $213$ MiB      & $3.74$ & $0.64$  & $0.7$ & $155.11$ & $212.98$ \\	
			\hline
		\end{tabular}
	\end{center}
}
	\caption{The \textit{wall clock} time---the amount of time that elapsed from the start to the completion of the instance---is given as microseconds per input base. The \textit{efficiency} column provides the CPU efficiency values, i.e. the proportion of time for which the CPU was occupied and not waiting for I/O operations to finish, as taken from the output of the \texttt{/usr/bin/time} command. The column \textit{\% mem} is the maximum value returned by \texttt{\%MEM} field of \texttt{top} command representing the percentage of used internal memory. The column \textit{aux disk space} denotes the maximum amount of external memory in GiB used to accommodate all the working files, \textit{tot disk space} includes the output files too.
	}
	\label{tab:compare_egsa}
\end{table}

Actually, we were not able to compare the performance of \BCRLCP\ and \EGSA\ on very large datasets. Indeed, current implementation of \EGSA\ shows an overhead in disk space usage that prevented us to make tests on very large string collections. For instance, for a collection of $400$ millions of strings of length $100$ \EGSA\ needs $40$ bytes per symbol, whereas \BCRLCP\ uses $5$ byte per symbol. So, the datasets in Table \ref{tab:inputInstances} are too demanding to be managed.
However, experiments have been conducted using smaller real collections (see Table \ref{tab:datasets}): the dataset $D_1$ (long jump library of Human Chromosome 14\footnote{\texttt{http://gage.cbcb.umd.edu/data/index.html}}), the datasets $D_2$ and $D_3$ (Human Genome sequences from Sequence Read Archive\footnote{\texttt{http://www.ebi.ac.uk/ena/data/view/ERR024163}}) are collections of DNA strings on the alphabet $\{A, C, G, T, N\}$.
Moreover, in order to evaluate \BCRLCP\ and \EGSA\ when the dataset contains strings over alphabets with more than five symbols, we have created a new collection of proteins (called $D_4$), obtained from UniProtKB/TrEMBL dataset, by truncating the strings longer than $100$ aminoacids, so that it has the same maximum string length as the datasets $D_1$, $D_2$, $D_3$.

For each dataset, \BCRLCP\ is compared with \EGSA\ using two different values of \MEMLIM: an higher value is comparable with the size of memory required from the data structures used by \BCRLCP; a lower value have been chosen to produce a number of subsets less than $1,024$. We have not set other parameters of \EGSA\ code.

The results of our tests are described in Table \ref{tab:compare_egsa}. The first three experiments, in which we consider strings over the same alphabet, show that the bigger the string collection is, the better is the overall performance of \BCRLCP\ with respect to \EGSA. In fact, the wall clock time values of \BCRLCP\ and \EGSA\ become comparable, but the total disk usage for \EGSA\ significantly increases. However, the fourth experiment in Table \ref{tab:compare_egsa} shows that the alphabet size is a parameter that seems to have a significant impact on the wall clock time value for \BCRLCP\  rather than for \EGSA.

Finally, we observe that, as pointed in Remark \ref{rm:GSA}, \BCRLCP\ could produce as additional output the generalized suffix array.  Although our algorithm is not optimized for this purpose, our tests have shown that the wall clock time is about twice because the I/O operations are doubled. However, the computation of the GSA does not produce an increase in the internal memory.

\section{Conclusions}\label{sec:conclusion}

In this paper, we proposed \BCRLCP\ which is a lightweight algorithm to construct, at the same time, the \LCP\ array  and the \EBWT\ of a collection of strings.

Actually, the \LCP\  and \BWT\ are two of the three data structures needed to build a compressed suffix tree (CST) \cite{Sadakane2007} of a string.
The strategy proposed in this paper could enable the lightweight construction of CSTs of string collections for comparing, indexing and assembling vast datasets of strings when memory is the main bottleneck.

The algorithm \BCRLCP\ is a tool parameter-free designed for very large collection, indeed Table \ref{tab:inputInstances} shows that it also works with collections of $74$GiB and Tables \ref{tab:varlenres} and \ref{tab:compare_egsa} show that it can be also used on large collections with varying string lengths.

The experimental results show that our algorithm is a competitive tool for the lightweight simultaneous computation of \LCP\ and \EBWT\ on string collections.

Our current prototype can be further optimized in terms of memory by performing the sorting step in external memory. Further saving of the working space could be obtained if we embody our strategy in \EMBWT or \EMBWTpp (see \cite{BauerCoxRosoneTCS2013}). These methods, although slower than \BCR, need to store only a constant and (for the DNA alphabet) negligibly small number of integers in RAM regardless of the size of the input data.

\section{Acknoledgements}
G. Rosone and M. Sciortino are partially supported by the project MIUR-SIR CMACBioSeq
(``Combinatorial methods for analysis and compression of biological sequences'') grant n.~RBSI146R5L
and by ``Gruppo Nazionale per il Calcolo Scientifico (GNCS-INDAM)''.
A. J. Cox is an employee of Illumina Cambridge Ltd. \\
The authors are very grateful to the anonymous referees for their helpful remarks and constructive comments.

\bibliographystyle{amsplain}


\end{document}